\documentclass[preprint]{iacrtrans}
\usepackage[utf8]{inputenc}

\usepackage{xcolor}
\usepackage{extarrows}
\usepackage{braket}
\usepackage{soul}
\let\underline\ULunderline
\usepackage{graphicx}
\usepackage{tikz}
\usepackage{tabularx}
\usepackage{stmaryrd}
\usepackage{comment}
\usepackage{dsfont}
\usepackage{tcolorbox}
\usepackage{multirow}

\usepackage{mdframed}
\usepackage{dirtytalk}
\usepackage{algorithm}
\usepackage{algorithmic}
\usepackage{graphicx}
\usepackage{float}
\usepackage{enumitem}
\usepackage{multicol}
\usepackage{url}
\usepackage{booktabs}

\usepackage{hyperref}
\usepackage{cleveref}
\usepackage[normalem]{ulem}
\definecolor{LimeGreen}{rgb}{0.2, 0.8, 0.2}
\usepackage{todonotes}

\newcommand{\eps}{\epsilon}
\newcommand{\I}{\mathcal{I}}

\newcommand{\C}{\mathcal{C}_{auth}}

\newcommand{\Z}{\mathbf{Z}}
\newcommand{\QKD}{\mathcal{QKD'}}
\newcommand{\QL}{\mathcal{QL}}
\newcommand{\SD}{\mathcal{SD}}
\newcommand{\PP}{\mathcal{PP}}
\newcommand{\A}{\mathcal{A}}

\newcommand{\T}{\mathcal{T}}
\newcommand{\U}{\mathcal{U}}
\newcommand{\R}{\mathcal{R}}
\renewcommand{\S}{\mathcal{S}}
\newcommand{\Ho}{\mathcal{H}}

\newcommand{\D}{\mathcal{D}}

\newcommand{\M}{\mathcal{M}}
\renewcommand{\P}{\mathcal{P}}

\newcommand{\bs}{\setminus}
\newcommand{\SIM}{\mathcal{SIM}}

\newcommand\numberthis{\addtocounter{equation}{1}\tag{\theequation}}

\newcommand{\splitatcommas}[1]{%
    \begingroup
    \begingroup\lccode`~=`, \lowercase{\endgroup
        \edef~{\mathchar\the\mathcode`, \penalty0
            \noexpand\hspace{0pt plus 1em}}%
        }\mathcode`,="8000 #1%
    \endgroup
}

\newtheorem{assumption}{Assumption}

\newcounter{protocol}

\begin{document}
\title[Security of a secret sharing protocol on the Qline]{Security of a secret sharing protocol on the Qline}

 \author{Alex B. Grilo$^1$ \and Lucas Hanouz$^{1,2,*}$ 
 \and Anne Marin$^2$ }
 \institute{
 Sorbonne Université, CNRS, LIP6, Paris, France
 \and
 VeriQloud, Paris, France
 }
\maketitle

\vspace{-1.2cm}
\begin{center}
    \small{$^{*}$ Corresponding author: lucas.hanouz@lip6.fr}
\end{center}


\begin{abstract}
Secret sharing is a fundamental primitive in cryptography, and it can be achieved even with perfect security. However, the distribution of shares requires computational assumptions, which can compromise the overall security of the protocol. While traditional Quantum Key Distribution (QKD) can maintain security, its widespread deployment in general networks would incur prohibitive costs.

In this work, we present a quantum protocol for distributing additive secret sharing of 0, which we prove to be composably secure within the Abstract Cryptography framework. 
Moreover, our protocol targets the Qline, a recently proposed quantum network architecture designed to simplify and reduce the cost of quantum communication.
Once the shares are distributed, they can be used to securely perform a wide range of cryptographic tasks, including standard additive secret sharing, anonymous veto, and symmetric key establishment.
\end{abstract}

\section{Introduction}

Secret sharing is a fundamental cryptographic primitive that enables the sharing of a secret among multiple parties, such that only specific predefined sets of shares allow to recover the original secret, while any other sets give no information about it. 

While classical secret sharing protocols, such as those introduced by Shamir and Blakley \cite{S79,B79}, achieve perfect (information-theoretic) security, the effective deployment of secret sharing protocols faces a critical challenge: the shares must be securely distributed to the participants, ensuring their privacy and~integrity.

In most cases, the share distribution is secured using standard cryptographic approaches such as public-key encryption, which lowers the overall security of the scheme. More concretely, using classical cryptography to distribute the shares, the security of data transmission holds under structured computational assumptions and cannot achieve the information-theoretic guarantees that secret sharing schemes provide.

This limitation, however, can be circumvented if we consider quantum cryptography. It is now well known that Quantum Key Distribution (QKD) protocols enable the establishment of a secure communication channel from an authenticated classical channel by leveraging the fundamental principles of quantum mechanics \cite{BB84,TL17,Portman_Renner14,BF12}. 
However, while QKD enhances security, its practical implementation introduces significant challenges. Establishing secure channels using QKD requires expensive quantum hardware and suffers from low efficiency. Moreover, scaling such systems to securely distribute the shares of a secret sharing scheme would be prohibitively inefficient, as both the infrastructure needed to support QKD and the required amount of secure communication grows rapidly with the number of participants and communication links.

Recently, a quantum network architecture called \emph{Qline} has been introduced \cite{doosti-hanouz23establishing} as an attempt to increase the connectivity of QKD networks, reduce their costs and improve their accessibility to end-users. The Qline consists of a standard QKD setup where a single qubit source and detector are linked, but with intermediate nodes added in between, which only have the ability to perform single-qubit rotations---a task that can be implemented with much cheaper devices. We will call the nodes of the Qline \emph{players} in this work. 

Despite having a simpler setup, it has been recently shown that Qline enables several interesting cryptographic protocols. Clementi \textit{et al.} demonstrated a Quantum-enhanced Classical multiparty computation protocol on the Qline \cite{CPEW17}. Later, Doosti et al. \cite{doosti-hanouz23establishing} showed that any pair of player can establish symmetric keys with the same level of security as QKD with the help of trusted end nodes, and Polacchi et al. \cite{polacchi2023multi} introduced a protocol for secure multi-client delegated quantum computing for a Qline connected to a quantum computer. In all of these protocols, the main idea is to use Qline to allow a pair of nodes to perform a secure operation (such as communication or computation). We notice that while this can be used to distribute shares of a secret sharing scheme with information-theoretic security more cost-effectively compared to pairwise QKD, it still suffers from a linear overhead on the number of shares for it.

Our main result is to show that additive secret sharing can directly and securely be performed on the Qline without scaling overhead on the shares. The main novelty is to exploit the {\em global} correlations that Qline provides us to achieve additive secret sharing of $0$ (i.e. the message is fixed).
We notice that previous works have introduced propositions exploiting this idea\cite{old_SSqline_experimental,Differential-phase-shift}, but they lack a security proof against general attacks, composable, and under the most general dishonest participant scenarios. On the other hand, we achieve a protocol that we prove to be secure in the composable framework of Abstract Cryptography\cite{AbstractCrypto}, while preserving the benefits of the Qline architecture regarding simplicity and cost of implementation.

In short, our protocol works as follows. The first player of the Qline sends a random BB84 states\footnote{In fact, we use the Hadamard Basis and the circular basis, but we prefer to continue the rough exposition with more well-known states.}, each intermediate player re-randomizes the states, and the final player chooses to measure the received qubits in the Hadamard or computational basis uniformly at random. Then, the players perform a classical protocol to check the integrity of the shares and to correct any error incurred by the noise of quantum devices. In order to prove the security of our protocol, we require that at least $2$ players are honest (which is natural for additive secret sharing of $0$), and that the players share a classical authenticated channel with random subset broadcast (see \Cref{sec:discussion} for a formal definition and a discussion on how to implement it). The core of the technical contribution is to show that our protocol is secure in a composable way.

This work was primarily motivated by the recent implementation of a Qline at VeriQloud, Paris, France. Our protocol is specifically designed to be compatible with their architecture, and simulations indicate that the sharing between four participants of a $2$ Mbit secret can be expected to be achieved in less than $5$ minutes on their setup.

To illustrate the protocol’s performance, we compare it in \Cref{tab:perf_comp} with the following alternatives of using classical secret sharing, along with either QKD or Qline's key establishment to distribute the shares.
In the following table, for each of these alternatives and depending on the number $J$ of players, we show the \emph{cost} (hardware requirement) of an architecture allowing any player to share a secret, and the \emph{efficiency}, measured in the number of required qubits to transmit to share one secret.
We use realistic and identical parameters and targeted metrics.

\begin{table}[h] 
\centering
        \begin{tabular}{l | c | c | c}
            \toprule
             & \textbf{\small QKD + classical} & \textbf{\small Qline + classical} & \textbf{\small Our protocol}

             \\
                    
            \textbf{} & \textbf{\small secret sharing} &\textbf{\small secret sharing} & \textbf{} \\
  \hline
            \textbf{\small Cost} & & & \\
            {\small Number of quantum channels} & $J^2$ & $1$ & $1$ 
            \\
            \hline
\textbf{\small Efficiency} & & & \\
            {\small Number of qubits to receive} & $J \times 10^7$ & $J \times 10^7$ & $10^7$ \\
                \bottomrule
        \end{tabular}
    \caption{Comparison of solutions with 1.7 Mbits share size and $ 10^{-11}$ distinguishing~advantage} \label{tab:perf_comp}
\end{table}

\paragraph{Applications.} We describe now some of the applications of the primitive that we implement, i.e. secret sharing of the bit $0$, when used along with classical authenticated communication. 
First, we notice that we can implement standard additive secret sharing where the dealer chooses a secret bit string to share instead of having this value fixed to $0$. To achieve such a primitive from shares of $0$, the dealer can publish the one-time-pad encryption of his secret with his share as the key. 

Another application is anonymous veto, also known as the Dining Cryptographers Problem \cite{dining_cryptographers}, which is also the secure multi-party computation of the multiple-input boolean OR function. To achieve anonymous veto from $n$ sharings of $0$, the players perform $n$ rounds of announcement with different announcement orders such that each player is last in one round. for all rounds, following the corresponding announcement order, the players broadcast either their share, or a random string instead if they wish to veto. For each round, the sum of the announcements is then compared to the all $0$ string: inequality shows that at least one player vetoed. See \cite{BT07} for a similar construction.

Finally, symmetric key establishment can be achieved by asking all players but two to reveal their share and having one of the two remaining players XOR it's share with all the thereby-revealed ones. Previous works\cite{doosti-hanouz23establishing} already introduced the corresponding protocol, but their security proof requires the honest collaboration of the end nodes of the network. We discuss in \Cref{sec:discussion} how, for this particular application, our proof amounts to the same result, \emph{without} this trust assumption.

Remarkably, for these applications, our protocol can be run in an offline phase, to then only in a later online phase, decide the cryptographic task to perform along with the set of involved players and use the shares together with classical authenticated communications to securely produce the desired resources at a high bit rate. This opportunity is all the more meaningful when considering the slow rates imposed by current quantum hardware.\\

The remainder of the manuscript proceeds as follows. We introduce preliminary information in \Cref{sec:prelim}. We then present our assumptions in \Cref{sec:assumptions}, our protocol in \Cref{sec:protocol}, and it's security proof in \Cref{sec:security}.

\subsection*{Acknowledgments} \label{par:acknowledgments}

\hspace*{\parindent}We would like to thank Georg Harder and Anthony Leverrier for their valuable assistance regarding the question of syndrome leakage. We thank Céline Chevalier for her guiding insights on technical parts of the proof.

ABG is supported by the European Union's Horizon Europe Framework Program under the Marie Sklodowska Curie Grant No. 101072637, Project Quantum-Safe Internet (QSI).  This work is part of HQI initiative (www.hqi.fr) and is supported by France 2030 under the French National
Research Agency award number ANR-22-PNCQ-0002. This work was funded by the European Union's Horizon Europe research and innovation program under grant agreement No. 101102140 – QIA Phase 1.

\section{Preliminaries} \label{sec:prelim}

We recall the notation of basic concepts of quantum information theory in \Cref{sec:quantum-information}. For a more detailed introduction to the topic, we refer to~\cite{nielsen&chuang}.
In \Cref{sec:AC_presentation}, we present the abstract cryptography framework. Finally, in \Cref{sec:def-qline}, we review the Qline architecture.

We defer to \Cref{sec:notations} for a summary of the notation used throughout this paper.

    \subsection{Quantum information theory}\label{sec:quantum-information}
        We assume basic knowledge about the theory of quantum communication and computing.
        
        We denote the eigenstates of the Hadamard basis by $\ket{+} = \frac1{\sqrt2}\Big(\ket{0} + \ket{1}\Big)$ and $\ket{-} = \frac1{\sqrt2}\Big(\ket{0} - \ket{1}\Big)$. The classical outcome of a measurement in the Hadamard basis yields $0$ if $\ket{+}$ is measured and $1$ if $\ket{-}$ is measured.

        We denote the eigenstates of the circular basis by $\ket{+_i} = \frac1{\sqrt2}\Big(\ket{0} + i\ket{1}\Big)$ and $\ket{-_i} = \frac1{\sqrt2}\Big(\ket{0} - i\ket{1}\Big)$ with $i^2 = -1$. By convention in this paper, we consider that the classical outcome of a measurement in the circular basis yields $0$ if $\ket{-_i}$ is measured and $1$ if $\ket{+_i}$ is measured. This mismatch in the notation between Hadamard and circular basis will improve the clarity of later equations.

        We denote the Pauli $\Z$ gate
          $\begin{bmatrix}
              1&0\\
              0&-1
          \end{bmatrix}$, $\Z^\frac 12$ being the phase gate 
          $\begin{bmatrix}
              1&0\\
              0&i
          \end{bmatrix}$ \vspace{4mm}.

        For a state $\sigma_R$ on registers $R$ and $S$, $Tr_R(\sigma_S)$ denotes the state obtained by tracing out the register $R$.
        
        The trace distance $Tr|\sigma - \gamma|$ is a measure of the distinguishability between two states $\sigma$ and $\gamma$. We write $\sigma \approx_\eps \gamma$ when $Tr|\sigma - \gamma| \leq \eps$

        Throughout this article, we use the term single-qubit state to denote a two-dimensional, potentially mixed, state.

    \subsection{The abstract cryptography framework} \label{sec:AC_presentation}
    
        We prove the security of our protocol using the \emph{Abstract Cryptography} framework~\cite{AbstractCrypto}. This framework is designed to guarantee the composability of the security of cryptographic constructions while remaining as general as possible concerning security notions. 
        
        In this framework, cryptographic protocols are defined as systems: abstract objects with \emph{interfaces} that define all possible inputs and outputs of the said system. Each interface represents an entity's access to the system. A cryptographic construction typically includes \emph{player} interfaces, where the interactions between the honest players and the system occur, as well as an adversarial interface, called the \emph{outer} interface, which encapsulates the attacker's capabilities.
        
        Systems can be composed, either in parallel or sequentially.
        The parallel composition of two systems $\R$ and $\S$, denoted $R||S$, is a system with the interfaces of both sub-systems. It simply describes the fact that these systems are put side by side and seen as a whole, unique system. The behavior of the composed system is naturally defined from the independent behaviors of the sub-systems (c.f \Cref{fig:def-composition}).
        
        \begin{figure}[!ht]
            \centering 
            \includegraphics[scale=0.50]{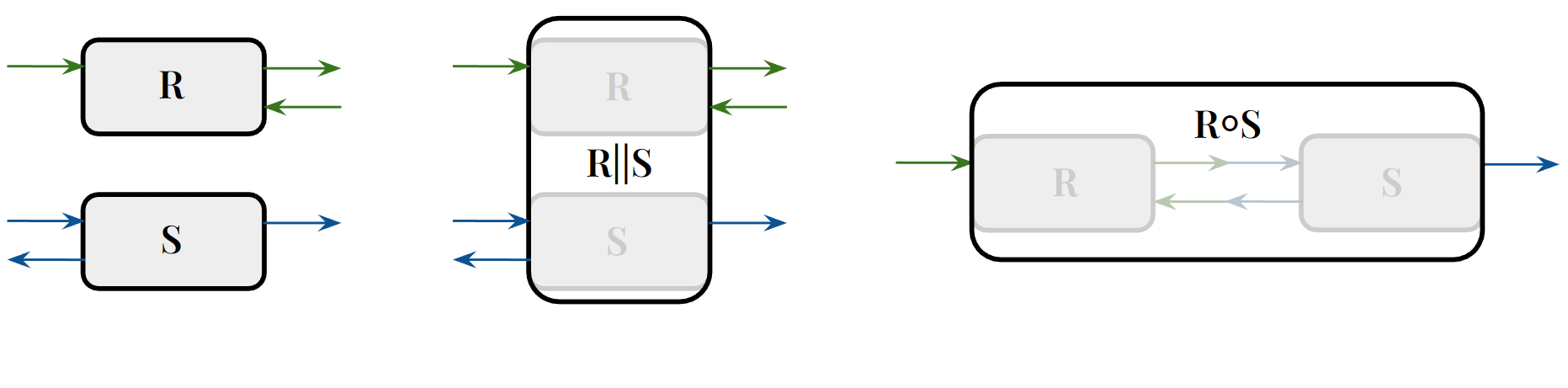}
            \caption{\centering  Composition of abstract systems}
            \label{fig:def-composition}
        \end{figure}
        
        The sequential composition describes the fact that the output of a system can be used as input by other systems. For instance, two systems $\R$ and $\S$ can be sequentially composed \emph{at the interfaces $i_\R$ of $\R$ and $j_\S$ of $\S$} if each input (respectively output) of these interfaces can be associated with a unique output (respectively input) of the other interface. When it is clear at which interfaces a sequential composition occurs, we denote it  $R\circ S$ or simply $RS$ without specifying the interfaces $i_R$ and $j_S$. The resulting system has all the interfaces of both sub-systems except from $i_R$ and $j_S$. %
        
        In this framework, the security of a cryptographic scheme is defined as the ``closeness'' of that system to an ideal version of it. In this work, this closeness is measured using the \emph{distinguishing} pseudo-metric (as in QKD security proofs~\cite{rennerabstract}), which is defined as the maximum \emph{distinguishing advantage} on two systems, over all computationally unbounded entities (called \emph{distinguishers}).
        The distinguishing advantage of a distinguisher on two sets of signals (inputs and/or outputs) is the value $\eps$ such that when given either the first set or the second one with equal probability $\frac{1}{2}$, the distinguisher succeeds in guessing which one it is with probability $\frac{1}{2} + \eps$.
        The distinguishing advantage on two systems $P$ and $\tilde P$ is the distinguishing advantage on their inputs and outputs.

        We wright $P \approx_\eps \tilde P$ when the distance (measured by the distinguishing pseudo-metric) between the systems (or signals) $P$ and $\tilde P$ is no more than $\eps$.
        
        Formally, A protocol $P$ of ideal version $\tilde P$ is said to be \emph{$\eps$-secure} if there exists a system $\SIM$ called \emph{simulator} such that $P \approx_\eps \tilde P\circ \SIM$.
        
        The distinguishing pseudo-metric leads to a composable definition of security (Theorem~1 of~\cite{AbstractCrypto}), meaning that the composition of an $\eps_1$-secure and an $\eps_2$-secure system is always $(\eps_1+\eps_2)$-secure.

    \subsection{The Qline Architecture} \label{sec:def-qline}
        A Qline consists in an initial node that can generate a given range of qubit states, an arbitrary number of intermediate nodes that can apply certain single qubit operations to these qubits, and a final node that can measure them in a chosen basis. An example of a Qline with four players is depicted in \Cref{fig:Qline}.

        \begin{figure}[ht] \label{fig:Qline}
        \begin{center}
        \begin{tikzpicture}
        	\fill[gray] 	(-0.5,0) arc (270:90:0.5cm) ;
        	\fill[gray]	(-0.6,0) rectangle (0,1);
        	\draw	(0,0.5) -- (1.5,0.5) ;
        	\fill[gray] 	(1.5,0) rectangle (2.5,1);
        	\draw	(2.5,0.5) -- (4,0.5) ;
        	\fill[gray] 	(4,0) rectangle (5,1);
        	\draw	(5,0.5) -- (6.5,0.5) ;
        	\fill[gray]	(6.5,0) rectangle (7.1,1);
        	\fill[gray] 	(7,1) arc (90:-90:0.5cm) ;
        	\node at (-0.5,-0.3) {Player 1};
        	\node at (2,-0.3) {Player 2};
        	\node at (4.5,-0.3) {Player 3};
        	\node at (7,-0.3) {Player 4};
        \end{tikzpicture}
        \caption{A Qline with four players}
        \end{center}
        \end{figure}
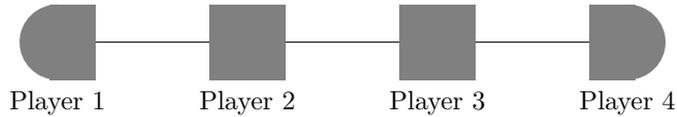

        In this work, we consider a Qline with the following properties:
        \begin{itemize}
        \item The first node can generate and send the four following states: $\ket{+}$, $\ket{-}$, $\ket{+_i}$, $\ket{-_i}$.
        \item The intermediate nodes can apply the $\Z^x$ operation to single qubit states, with $x\in\{0, \frac12, 1, \frac32\}$. 
        \item The last node can measure single qubit states in either the Hadamard or the circular basis.
        \end{itemize}

\section{Adversarial model}
\subsection{Assumptions} \label{sec:assumptions}
    We study the security of the protocol under an \emph{active}, \emph{unbounded}, and \emph{participant} adversarial model. This means that we consider that the adversary has access to noiseless, unbounded, quantum and classical computational power and storage (\emph{unbounded}), that they can attack the protocol during its execution (\emph{active}), and most importantly that they can corrupt parties involved in the protocol, meaning that they take complete control over their knowledge and behavior (\emph{participant}). An uncorrupted player is said to be \emph{honest}. 
    
    In order to prove the security of the protocol, we require the following assumptions.
    \begin{assumption} \label{ass:2_pl_honest}
        At least two players are honest.
    \end{assumption}
    Note that this is a minimal assumption for our use case. Having only one honest player would not achieve any interesting result as the goal of the protocol is precisely that the secret of any player can be recovered using all the other's secrets.

    \begin{assumption} \label{ass:Urandom}
        Perfect randomness: Each player has access to an independent uniform random number generator.
    \end{assumption}

    \begin{assumption} \label{ass:Sealed_labs}
        Sealed laboratories: No unwanted information transfer occurs at the frontier of the honest players laboratories.
    \end{assumption}
    This is arguably the most challenging assumption to ensure in practice. It prevents side-channel attacks, which are inherently difficult to defend against. In particular, this assumption also encompasses that honest players receive and send only single qubit states, meaning that no higher-dimensional states arise when expecting two-dimensional ones.

    \begin{assumption} \label{ass:CAC}
        Classical authenticated channel with random subset broadcast.
    \end{assumption}
    We describe \Cref{ass:CAC} in more details. The players are assumed to have access to a classical communication channel, which allows them to broadcast classical messages to all the other players, while ensuring the three following main features:
    \begin{enumerate}
        \item \textbf{Authentication:} The messages going through this channel cannot be tampered with, and come with the identity of the sender.
        \item \textbf{Random subset broadcast:}
        Whenever required, the players can perform a \emph{random subset broadcast} over the channel. In this procedure, each player first inputs an ordered list of values, and, in a second stage, the procedure randomly samples a subset of indices for which the corresponding values of all players are revealed and broadcast to everyone. See \Cref{sec:discussion} for more details.
        \item \textbf{Distributed coin-flipping:} Whenever required, the players can perform a \emph{distributed coin-flipping} over the channel. This procedure allows the players to agree on a value that is uniformly random and independently sampled. The procedure can abort.
    \end{enumerate}
    
    Such a channel is required to prevent potential malicious players to cheat by making choices that depend on the other player's announcements.
    We further discuss this assumption in \Cref{sec:discussion} and provide constructions of such a channel using standard assumptions.\\

\subsection{Discussion on Assumption \ref{ass:CAC}} \label{sec:discussion}
    \Cref{ass:CAC} and in particular the aspect of \emph{random subset broadcast} is rather specific to our work and is unconventional in cryptography. While it could be replaced with other standard computational assumptions, we chose to retain \Cref{ass:CAC} as we believe it best captures the purpose and significance of the assumption while highlighting the key protocol components that rely on it for security.
    In this section, we discuss different approaches to satisfy \Cref{ass:CAC} based on more common assumptions.

    \textbf{Authenticated channel:}
    Authenticated channels are well-studied resources that can be obtained from many different cryptographic solutions and for different paradigms \cite{J90, WegmanCarter81, SPHINCS+, NTRU}. Remarkably, using pre-shared keys, authenticated channels can be obtained information-theoretically. It is a required assumption in Quantum Key Distribution protocols \cite{B79,BB84}.

    \textbf{Distributed coin-flipping:}
    \emph{distributed coin-flipping} is a protocol in which the participants agree on a random value with the guarantee that the probability distribution of the outcome is uniform, no matter what an adversary tries to do. We refer to \cite{distr_coin_flip} for a thorough study of distributed coin-flipping protocols and the required assumptions in our setup.
    
    A particular situation that appears interesting enough to be mentioned is when one honest player is identified. This situation can arise when it has already been decided what the shares are going to be used for. If the goal is for instance to use the shares for the sharing of a document, the document holder is by definition honest. In this case, a simple implementation of the coin-flipping that does not involve any further assumption is to make the honest player sample the random value and simply announce it to the others.

    \textbf{Random subset broadcast:}
    A random subset broadcast is a procedure that breaks down in two stages. In the first stage, each player chooses (and delivers to the procedure) a list of values indexed by a given set $S$ and a unique size $s$ is chosen. In the second stage, a subset of $S$ of size $s$ is randomly sampled and the procedure reveals to all players the values of each list that correspond to this subset, while the other values remain hidden.

    A perhaps quite natural construction of such a procedure would be to use a commitment scheme and to ask each player to commit on each of his values in the first stage, to then in the second stage perform a distributed coin flip to randomly sample the subset, and finally to open their commitments of the required values.
    Commitment schemes rely on the assumption of the existence of a One-Way Function \cite{Naor}. While this assumption is equivalent to the security of secret key encryption and most currently used implementations of One-way functions are widely believed to be secure even against quantum computers (\cite{C24sponge}), relying on one-way functions bottlenecks the security. One must however notice that (if using statistically hiding commitments) the binding of the commitment scheme is only required during the time of the procedure for the final shares to be information-theoretic secure. This property is called \emph{everlasting security} and is highly desirable when seeking long-term security.

    A very particular, yet relevant situation is when all honest players are identified. For instance, this occurs when it has already been decided that the shares resulting from the protocol will be used to establish secret symmetric keys between two fixed players, as studied in \cite{doosti-hanouz23establishing}. In this specific case, a simple implementation of the random subset broadcast procedure that does not involve any computational assumption is as follows: First, all dishonest players broadcast all of their values. Then a honest player samples and announces the subset, and all honest players announce only their values that correspond to the subset.

\section{The Protocol} \label{sec:protocol}

    In this section we present a quantum-assisted secret sharing protocol that is supported by the Qline architecture described in \Cref{sec:def-qline}. 
   This is the \emph{prepare-and-measure} version of the protocol
   and we will introduce the entanglement-based version later.

    The protocol involves $J$ players, each having exclusive control of one device of the Qline architecture. The players are named according to \Cref{sec:def-qline}.

    The protocol can be decomposed into two major steps: the \emph{State distribution} step involving quantum communication and the \emph{Post-processing} step which only requires classical computation and communication.
    
\tcbset{colframe=black!80, colback=gray!10, boxrule=0.5mm, arc=3mm, fonttitle=\large}

\begin{tcolorbox}[title=Parameters]
        
        The protocol is parameterized by the following variables:
        \begin{enumerate}[noitemsep,topsep=0pt]
            \item A security parameter $N$.
            \item A correctness parameter $\eta$.
            \item An integer $\tau'$ such that $\tau' = \omega(log(N))$ and $\tau'=o(N)$.
        \end{enumerate}
        
\end{tcolorbox}

\begin{tcolorbox}[title=State distribution (prepare-and-measure version)]

        \begin{enumerate}[noitemsep,topsep=0pt]
            \item Each player $j\in[J]$ samples $2N$ random bits $(b_n^j)_{n \in [N]}$ and $(v_n^j)_{n \in [N]}$ and computes $x_n^j = \frac{b_n^j}{2} + v_n^j$, for all $n\in[N]$.
            \item Player $1$ generates the state $\ket{\Phi^1} = \bigotimes\limits_{n \in [N]} \Z^{x^1_n}\ket{+}$ and sends it to player $2$.
            \item Players $j\in\{2, ... J-1\}$ receive a state $\tilde\Phi^{j-1}$ from player $j-1$, apply the operation $U^j = \bigotimes\limits_{n\in[N]} \Z^{x^j_n}$ to it 
            and send the resulting state $\Phi^j = U^j\tilde\Phi^{j-1}U^{j\dagger}$ to player $j+1$.
            \item Player $J$ receives $\tilde\Phi^{J-1}$ from player $J-1$, and measures each qubit $n\in[N]$ in either the Hadamard basis if $b_n^j = 0$ or in the circular basis if $b_n^j = 1$. The classical outcome of these measurements are denoted $(v_n^j)_{n\in [N]}$.
        \end{enumerate}
\end{tcolorbox}

\begin{tcolorbox}[title=Post-processing (part 1)]

 \label{sec:post-processing}
    At any point, if the classical channel fails, the protocol aborts.
        \begin{enumerate}[noitemsep,topsep=0pt]

        \item \textbf{Announcements:}
            \begin{enumerate}[label=\arabic{enumi}.{\arabic*}.,noitemsep,topsep=0pt]
                \item The players perform the first stage of two \emph{random subset broadcast} procedures with respectively their basis choices $b^j_n$ and values $v^j_{n}$ for $n \in [N]$, and then the second stages so that all the values $b^j_n$ for $n \in [N]$ are broadcast to all players, while a random subset $\T'$ of $[N]$ of size $\tau'$ is randomly sampled, and only $v^j_{n}$ for $n \in \T'$ are broadcast.
            \end{enumerate}
            
        \item \textbf{Sifting:}
            \begin{enumerate}[label=\arabic{enumi}.{\arabic*}.,noitemsep,topsep=0pt]
                \item The players compute the indices of the inconclusive rounds $\U = \Big\{n \in [N] :  \bigoplus\limits_{j \in [J]} b_n^j \neq 0 \Big\}$ and discard $b_n^j$ and $v_n^j$ where $n \in \U$. We define $L := N - |\U|$
                We keep the same notation for the remaining values, but adjust the indices of the rounds: 
                \begin{align*}
                    (b_n^j)_{n \in [L], j \in [J]} &:= (b_n^j)_{n \in [N]\bs \U, j \in [J]}\\
                    (v_n^j)_{n \in [L], j \in [J]} &:= (v_n^j)_{n \in [N]\bs \U, j \in [J]}\\
                \end{align*}
                We equivalently adjust the indices of $\T$ such that $\T \subset [L]$ (the rounds in $\T$ are still the rounds previously in $\T' \bs \U$)
            \end{enumerate}

    \end{enumerate}
\end{tcolorbox}

\begin{tcolorbox}[title=Post-processing (part 2)] 
    \begin{enumerate}[noitemsep,topsep=0pt, start=3]
    
        \item \textbf{Error estimation:}
            \begin{enumerate}[label=\arabic{enumi}.{\arabic*}.,noitemsep,topsep=0pt]
                \item The players compute the \emph{Qubit Error Rate}
                    \begin{equation} \label{eq:qber}
                    q = \frac{1}{\tau} 
                    \Big|\Big\{n \in \T :
                    \sum\limits_{j \in [J]} ( 2v^j_n + b^j_n ) = 2 \mod 4
                    \Big\}\Big|
                    \end{equation}
    
                    If $q > \delta $, the parties abort. 
                    
                    \item The players discard $b^j_{n}$ and $v^j_{n}$ for each index $n \in \T$. We define $M = L-\tau $ and again adjust the indices such that the remaining values are $(b_n^j)_{n \in [M], j \in [J]}$ and $(v_n^j)_{n \in [M], j \in [J]}$
            \end{enumerate}
    
        \item \textbf{Error correction:} 
            \begin{enumerate}[label=\arabic{enumi}.{\arabic*}.,noitemsep,topsep=0pt]
                \item Player $J$ updates his values $(v_n^j)_{n\in[M]}$ as
                \begin{equation*}
                    v_n^j := v_n^j \oplus \Big( \frac 12 \big(\sum\limits_{j\in[J]} b_n^j \mod 4\big) \Big)
                \end{equation*}
                \item The players agree on an error correction margin $\nu\in[0, \frac12 - q]$, as well as a \emph{linear} syndrome decoding protocol\footnotemark{} of correction rate $(q+\nu)$ that they will apply on their respective shares $v^{[J]}_{[M]}$.
                \item According to this syndrome decoding protocol, each player $j\in[J-1]$ computes and announces the syndrome $w^j$ of his share $v^J_{[M]}$.
                \item Player $J$ corrects it's share $v^J_{[M]}$ through the syndrome decoding protocol using $\bigoplus\limits_{j\in[J]}w^j$ as the correction syndrome.
                \item \textbf{Correctness check:} The players use the \emph{distributed coin flipping} procedure to randomly sample $f_{cc}$ from a $2$-universal family of \emph{linear} hash functions\footnotemark{} from $\{0, 1\}^{M}$ to $\{0, 1\}^{\eta}$. 
                Each player $j\in[J]$ computes and announces the hash $c^j = f_{cc}(v^j)$ and checks that
                \begin{equation}
                    \bigoplus_{j\in[J]} c^j = 0
                \end{equation}
                If the check fails, the protocol aborts.
            \end{enumerate}
            
        \item \textbf{Privacy amplification:}
            \begin{enumerate}[label=\arabic{enumi}.{\arabic*}.,noitemsep,topsep=0pt]
                \item The players agree on an integer $K<M$ and use the \emph{distributed coin flipping} procedure to sample a function $f_{pa}$ from a $2$-universal family of \emph{linear} hash functions from $\{0, 1\}^{M}$ to $\{0, 1\}^{K}$. 
                They compute their final share as 
                \begin{equation}
                    s^j = f_{pa}(v^j)
                \end{equation}
            \end{enumerate}
    \end{enumerate}

\end{tcolorbox}
\addtocounter{footnote}{-1}
\footnotetext[\value{footnote}]{\label{footnote:err_corr}By syndrome decoding, we refer to a protocol allowing one to compute the syndrome of a message, such that the combined knowledge of this syndrome and a noisy version of the message allows (efficient) computation of the original message. Such a syndrome decoding protocol comes with a \emph{correction rate} such that when the noise is less than or equal to that rate, the correction succeeds except with negligible probability in $N$. Linear syndrome decoding protocols can be derived from linear error correcting codes.}
\addtocounter{footnote}{1}
\footnotetext[\value{footnote}]{\label{ftnote1}Examples of such $2$-universal families of linear hash functions can be found in \cite{WC79}.}

    \subsection{Correctness} \label{sec:correctness}
    
    In this section, we prove the correctness of the protocol.
    For that, we prove in \Cref{prop:correctness} that when the protocol succeeds, the produced shares are correct with high probability. Then, in \Cref{prop:correctness2} we show that when the parties are honest the protocol successfully terminates if the noise in the devices is low enough. The correctness highly depends on the correctness parameter $\eta$ chosen by the players in the protocol during the correctness check step.
    \begin{proposition} \label{prop:correctness}
        Let $\eps_{cor} = 2^{-\eta}$.
        Assuming that the protocol successfully terminates, then with probability at least $1 - \eps_{cor}$, 
        \begin{equation} \label{eq:correctness}
            \bigoplus_{j\in[J]} s^j = 0.
        \end{equation}
    \end{proposition}

    \begin{proof}
        After the error correction step of a successful execution of the protocol, the correctness check verified that $\bigoplus\limits_{j\in[J]} f_{cc}(v^j) = 0$. Since $f_{cc}$ is sampled from a $2$-universal {\em linear} hash family, the probability that $\bigoplus\limits_{j\in[J]} v^j \ne 0$ is at most $2^{-\eta}$.
    \end{proof}

    \begin{proposition} \label{prop:correctness2}
    If the parties are honest and the depolarizing noise is $\mu<\delta$, then the protocol successfully terminates except with a probability at most negligible in $N$.
    \end{proposition}
    
    \begin{proof}
        The protocol may abort at 3 stages, and we bound each of these probabilities below.
        
        First, the parties would abort during sifting if $\tau < \frac{\tau'}{4}$. As the player's basis choices are uniformly random, the probability for each round $n\in[N]$ that $ \bigoplus_{j\in[J]}b_{n}^j = 0$ is $\frac 12$. Hence, by Hoeffding's bound, except with probability at most $e^{-\frac {\tau'}8}$, $\tau \geq \frac {\tau'}4$. Similarly, except with independent probability at most $e^{-\frac{N-\tau'}8}$, $M \geq \frac {N-\tau'}4$. Both happen with probability at least $p_1 < e^{-\frac{\tau'}{8}} + e^{-\frac{N-\tau'}{8}}$.
        
        Secondly, the parties abort during error estimation if $q > \delta $. During the state distribution step in an ideal noiseless case, for all $n\in[N]$, player $J$ is expected to measure the state
            \begin{align*}
                \Z^{\big(\sum\limits_{j\in [J-1]} v_n^j - \frac 12 b_n^j\big)} \ket{+}
                &= \Z^{\bigoplus\limits_{j\in [J-1]} v_n^j} \Z^{-\frac 12\sum\limits_{j\in [J-1]} b_n^j} \ket{+}\\
                &=  \begin{cases}
                       \Z^{\big(\bigoplus\limits_{j\in [J-1]} v_n^j\big)\oplus\big(\frac 12(\sum\limits_{j\in [J-1]} b_n^j \mod 4)\big)} \ket{+} &\quad\textit{if } \bigoplus\limits_{j\in[J-1]}b_n^j = 0\\
                       \Z^{\big(\bigoplus\limits_{j\in [J-1]} v_n^j\big) \oplus\big(\frac 12(1+\sum\limits_{j\in [J-1]} b_n^j \mod 4)\big)} \ket{+_i} &\quad\textit{if } \bigoplus\limits_{j\in[J-1]}b_n^j = 1\\
                    \end{cases} 
            \end{align*}
            where the first equality comes from the facts that $\Z^2 = I$  and $Z^{a+b} = Z^aZ^b$.
            
            As a consequence, for all $n\in[N]$ where player $J$ chooses the basis $b_n^J = \bigoplus\limits_{j\in[J-1]}b_n^j$ for his measurement, they should in principle obtain the following result deterministically: 
            \begin{equation} \label{eq:correctness_aux1}
                v_n^J = \Big(\bigoplus\limits_{j\in [J-1]} v_n^j\Big) \oplus\Big(\frac 12\big(b_n^J + \sum\limits_{j\in [J-1]} b_n^j \mod 4\big)\Big)
            \end{equation}

            Thus, because of the assumption on the noise and using Hoeffding's bound, the qubit error rate
            \begin{align*}
                q &= \frac{1}{\tau} 
                    \Big|\Big\{n \in \T :
                        \sum\limits_{j \in [J]} ( 2v^j_{n} + b^j_{n} ) = 2 \mod 4
                    \Big\}\Big| \\
                &= \frac{1}{\tau} 
                    \Big|\Big\{n \in \T :
                        \bigoplus\limits_{j \in [J]} b^j_{n} = 0 
                        \textit{ and }
                        2 v_{n}^j = 2 - \sum\limits_{j \in [J-1]} 2v^j_{n} - \sum\limits_{j \in [J]} b^j_{n} \mod4
                    \Big\}\Big| \\
                &= \frac{1}{\tau} 
                    \Big|\Big\{n \in \T :
                        \bigoplus\limits_{j \in [J]} b^j_{n} = 0 
                        \textit{ and }
                        v_{n}^j = 
                        1 
                        -
                        \big( \sum\limits_{j \in [J-1]} v^j_{n} \big) 
                        - 
                        \frac 12 \big( \sum\limits_{j \in [J]} b^j_{n} \mod4 \big)
                        \mod2
                    \Big\}\Big| \\
                &= \frac{1}{\tau} 
                    \Big|\Big\{n \in \T :
                        \bigoplus\limits_{j \in [J]} b^j_{n} = 0 
                        \textit{ and }
                        v_{n}^j \ne 
                        \big( \bigoplus\limits_{j \in [J-1]} v^j_{n} \big) 
                        \oplus 
                        \frac 12 \big( \sum\limits_{j \in [J]} b^j_{n} \mod4 \big)
                    \Big\}\Big|\\
                &= \frac{1}{\tau} 
                    \Big|\Big\{n \in \T :
                        b_{n}^J = \bigoplus\limits_{j\in[J-1]}b_{n}^j
                        \textit{ and }
                        \text{Equation \ref{eq:correctness_aux1} is invalidated}
                    \Big\}\Big|
            \end{align*} 
            is smaller than or equal to the threshold $\delta$, except with probability at most $p_2=e^{-2\tau(\delta-\mu)^2}$.
        
        Finally, the parties may abort at the correctness check. Note that after error estimation, for all $n\in[M]$, $b_n^J = \bigoplus\limits_{j\in[J-1]}b_n^j$ and thus \Cref{eq:correctness_aux1} is satisfied. After Player $J$ updated their value at the beginning of the error correction step, \Cref{eq:correctness_aux1} gives 
            \begin{equation} \label{eq:correctness_aux2}
                \bigoplus\limits_{j\in [J]} v_n^j = 0.
            \end{equation}
         Again by Hoeffding's bound and from the assumed bound on the noise, the Hamming weight of $\bigoplus\limits_{j\in [J]} v_n^j$ will be smaller than or equal to $q+\nu$ except with probability at most $p_3 = e^{-2M(q+\nu-\mu)^2}$. In this event, after the error correction step, due to the properties of the error correcting code (see \Cref{footnote:err_corr}), except with a negligible probability in $N$ that we denote $p_{ec}$, \Cref{eq:correctness_aux2} will strictly be satisfied for all rounds $n \in [M]$. Hence the correctness check will pass and the protocol will successfully terminate.

\medskip 

        To conclude, by the union bound, the protocol successfully terminates except with probability at most $p_1+p_2+p_3+p_{ec}$ which, taking into account the different minimum values of $\tau$, $q$ and $M$ under the assumed events, is lower than $e^{-\frac{1}{8}\tau'} + e^{-\frac{1}{8}(N-\tau')} + e^{-\frac {(\delta-\mu)^2}2\tau'} + e^{-\frac {(\nu-(\delta-\mu))^2}2(N-\tau')} +p_{ec}$ which is negligible in $N$ since $\tau'=\omega(log(N))$.
    \end{proof}

    \Cref{prop:correctness2} and \Cref{prop:correctness} together show the correctness of the protocol.

\section{Security} \label{sec:security}
    This section is dedicated to the proof of security of the protocol described in \Cref{sec:protocol}. We first introduce an entanglement-based version of the protocol in \Cref{sec:EB_state_distr} followed by formal definitions in \Cref{sec:def}. We establish the equivalence between the entanglement-based and the prepare-and-measure versions in \Cref{sec:PM=EB}, and then show the security of the entanglement-based version in \Cref{sec:EB_security}. Finally, we bring together the results to conclude the security proof in \Cref{sec:conclusion-security}.

    \subsection{Entanglement-based version of the protocol} \label{sec:EB_state_distr}

        The entanglement-based version of the protocol
        is identical to the prepare-and-measure version from \Cref{sec:protocol}, except for the state distribution step which is defined below. \footnote{The entanglement-based version requires the nodes of the Qline to have different capabilities compared to the prepare-and-measure version. As we only use the entanglement-based version as a tool to show the security of the prepare-and-measure version, this has no impact on the implementation requirements for the protocol.}

        \begin{tcolorbox}[title=State distribution (entanglement-based version)]
        
                \begin{enumerate}[noitemsep,topsep=0pt]
                    \item The players agree on a integer $N$. Each player $j\in[J]$ samples $N$ random bits $(b_n^j)_{n \in [N]}$
                    \item Player $1$ generates N copies of the state $\frac{1}{\sqrt{2}}(\ket{00}+\ket{11})$ and sends one qubit of each copy to player $2$.
                    \item Each player $j\in\{2, ..., J-1\}$, obtains $N$ qubits. 
                        For each qubit, player $j$ applies a CNOT gate with the latter qubit as the control qubit of the operation, and a freshly prepared qubit in the $\ket{0}$ state as the target qubit. Player $j$ then sends the first qubit to player $j+1$
                    \item Each player $j\in [J]$ measures each of their qubits $(\phi_{n}^j)_{n\in [N]}$ in either the Hadamard basis if $b_n^j = 0$ or in the circular basis if $b_n^j = 1$. The classical outcome of these measurements are denoted $(v_n^j)_{n\in [N]}$%
                \end{enumerate}
        \end{tcolorbox}
    \subsection{Definitions} \label{sec:def}

        We hereafter define the systems that are later used to prove the security of the protocol. This includes the systems $\QL_{EB}$ and $\QL_{PM}$ respectively implementing the entanglement-based and the prepare-and-measure versions of the protocol. We first define individual components in \Cref{sec:def_components}, and then the complete systems in \Cref{sec:def_systems}.%

        \subsubsection{Component systems}\label{sec:def_components}
            \begin{itemize}
                \item $\C$ is a $J$-player classical authenticated broadcast channel with a random subset broadcast procedure implementing \Cref{ass:CAC}.
                It provides each player, honest or dishonest, with the ability to broadcast messages to all the others while authenticating the source of the messages. This is modeled by $J$ player interfaces, each with an input $t^j$ for $j\in[J]$ and an output $t$ giving the transcript and the source of all the messages broadcast in the inputs of the other player interfaces. 
                $\C$ provides external entities with the ability to read the data or block the communications, but not to tamper with them. This is modeled by an outer interface providing as output no more than a copy of $t$, and receiving a binary input $\ell$, called a blocking lever, which, if set to '$1$', prevents the messages to pass through.
                $\C$ also provides the players with the ability to execute the \emph{random subset broadcast} and \emph{distributed coin-flip} procedures described \Cref{sec:discussion}.

                \item For each honest player $j\in\Ho$, $\SD_{PM}^j$ (respectively $\SD_{EB}^j$) is a system implementing the state distribution step for player $j$ of the prepare-and-measure version (respectively the entanglement-based version) of the protocol. It has an outer interface with a $N$-qubit state input $\rho^j_{in}$ and a $N$-qubits state output $\rho^j_{out}$, as well as an inner interface outputting the bits $b_{[N]}^j$ and $v_{[N]}^j$. 

                \item The post-processing system $\PP$ implements the post-processing step of the protocol described in \Cref{sec:protocol} (identical for any player $j\in\Ho$ and any version of the protocol). It has an inner interface with two $N$-bits inputs for $b_{[N]}^j$ and $v_{[N]}^j$, a player interface with a final share output $s^j$, as well as a side interface managing all the communications that occur on the classical authenticated channel. This side interface is designed to be plugged to a player interface of $\C$ and thus has an output $t^j$ for outgoing messages and an input $t$ for incoming classical communication.

                \item The systems $\P^j_{PM}$ (resp. $\P^j_{EB}$) implement the full prepare-and-measure (resp. entanglement-based) protocol described in \Cref{sec:protocol,sec:EB_state_distr} for a given honest player $j\in\Ho$. $\P^j_{PM}$ ($\P^j_{EB}$) is the sequential composition of the state distribution and the post-processing systems at their respective inner interfaces. It thus has the outer interface of $\SD_{PM}^j$ ($\SD_{EB}^j$) as well as the player and side interfaces of $\PP$. 
                \begin{align}
                    \label{eq:Pj_PM_def}
                    \P^j_{PM} &= \PP\circ\SD_{PM}^j \\
                    \label{eq:Pj_EB_def}
                    \P^j_{EB} &= \PP\circ\SD_{EB}^j
                \end{align}

            \end{itemize}
        \subsubsection{Complete systems} \label{sec:def_systems}
        We define here the main systems that describe the protocol and its security. These systems are represented \Cref{fig:Isigma}
            \begin{itemize}

                \item The system $\QL_{PM}$ (resp. $\QL_{EB}$) is a theoretical model of the Qline for the prepare-and-measure version (resp. entanglement-based version) of the protocol. It is composed of the systems $\P^\Ho_{PM}$ ($\P^\Ho_{EB}$) modeling the $H$ honest players, all composed sequentially to the classical channel $\C$. $\QL_{PM}$ ($\QL_{EB}$) has an outer interface composed of $\rho^\Ho_{in}$ and $\rho^\Ho_{out}$ the outer inputs and outputs of all honest players, as well as the unused signals of $\C$, namely the dishonest players inputs and outputs ($(t^j)_{j \in [J]\bs\Ho}$ and $J-H$ copies of $t$) and the blocking lever $\ell$. $\QL_{PM}$ ($\QL_{EB}$) also has a share interface with $s^\Ho$ the share outputs of the player interfaces of the honest player systems $\P^\Ho_{PM}$ ($\P^\Ho_{EB}$).
                
                The dishonest players are fully controlled by the outside environment and thus are not part of the system. Instead, all their inputs and outputs are exposed to the outer interface of $\QL_{PM}$ ($\QL_{EB}$), modeling the fact that the outside world has complete control over the inputs and absolute knowledge of the outputs. 

                According to these definitions, the systems $\QL_{PM}$ and $\QL_{EB}$ can equivalently be viewed as
                \begin{align} \label{eq:Ql_sys_def}
                    \QL_{EB} &= \C \circ ({\big|\big|}_{j\in\Ho} P^j_{EB})\\
                    \QL_{PM} &= \C \circ ({\big|\big|}_{j\in\Ho} P^j_{PM})
                \end{align}
                
                \item The ideal secret sharing system $\I$ has a \emph{share} interface and an \emph{outer} interface. The share interface has share outputs $(s^j)_{j\in\Ho}$, which are either binary strings of equal sizes, or the abort symbol $\bot$. The ideal property of $\I$ is captured by the fact that the only inputs and outputs of the outer interface, namely those exposed to the external entities, are the following:
                \begin{itemize}
                    \item A binary input $\ell$ called a a \emph{blocking lever}, which, if set to $1$, forces the system to abort regardless of any other input (enforcing $s^j=\bot$ for all $j \in [J]$).
                    \item An output $|s|$ giving the size of the honest shares outputs.
                    \item A "compromised" share $s^{compr}$ which is a binary string \emph{input} of size $|s|$. This models the fact that $\I$ handles dishonest entities: they are allowed to choose their share, and their sum (bit-wise xor $\oplus$), called $s^{compr}$, is taken into account in the honest shares' generation.
                \end{itemize}
                $\I$ guarantees that if $\ell$ is set to $0$, the shares are all uniformly random and independent of anything else, except for the last share $s^{j_H}$ which is given by 
                \begin{equation} \label{eq:sH}
                s^{j_H} = \Big(\bigoplus\limits_{j \in\Ho\bs\{{j_H}\}} s^j \Big) \oplus s^{compr}
                \end{equation}
                
                \item The system $\SIM$ is a simulator. It has an inner interface meant to connect to the outer interface of $\I$, involving the blocking lever $\ell$, the size $|s|$ of the honest shares, as well as the compromised share input $s^{compr}$.
                $\SIM$ also has an outer interface that matches the one of the $\QL_{EB}$ system. This interface consists of the following inputs and outputs:
                $%
                    \rho^\Ho_{in},
                    \rho^\Ho_{out}, 
                    t^{[J]\bs\Ho},
                    t, \ell
                $.
            
                The simulator is represented in \Cref{fig:Isigma}.
                In order to produce inputs and outputs of the outer interface, the simulator internally runs a copy of the $\QL_{EB}$ system and directly maps every input and output of it's outer interface to the one of $\SIM$. The share outputs of $\QL_{EB}$  however, labeled $s^\Ho_\SIM$, are used to compute the compromised share output of $\SIM$ as
                \begin{equation} \label{eq:scompr}
                    s^{compr} = \bigoplus\limits_{j\in \Ho} s^j_\SIM          
                \end{equation}

                Furthermore, in the event where $\QL_{EB}$ aborts (indicated by the output shares being set to $\bot$), the simulator will trigger the blocking lever $\ell$ of $\I$.
    
            \end{itemize}

        \begin{figure}[!ht]
            \centering \hspace*{0.1cm}
            \includegraphics[scale=0.365]{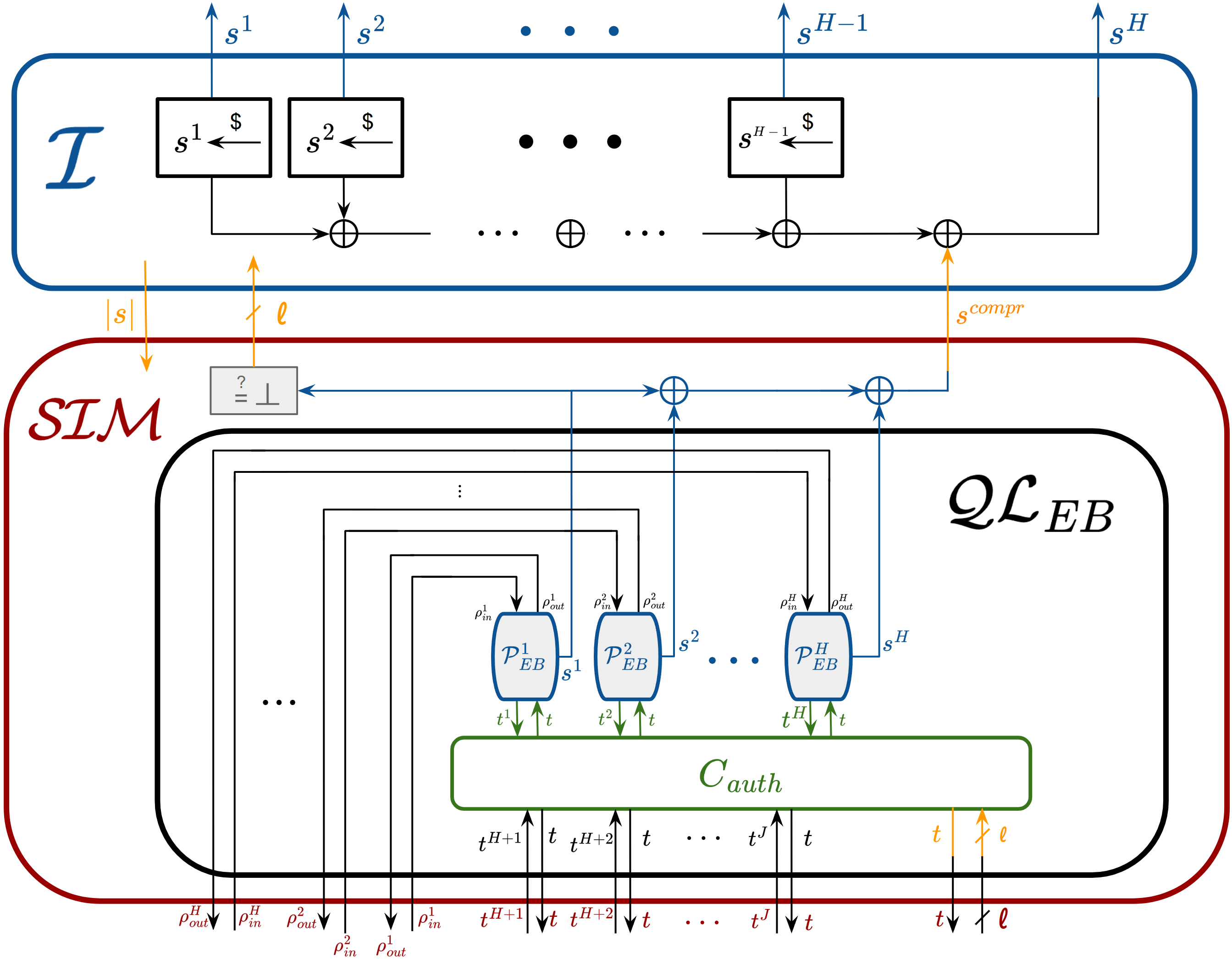}
            \caption{\centering The simulator $\SIM$ plugged on $\I$.}
            \label{fig:Isigma}
        \end{figure}

        \subsection{Equivalence between $\QL_{EB}$ and $\QL_{PM}$} \label{sec:PM=EB}
            This section is dedicated to prove the equivalence of the entanglement-based version of the protocol implemented by the system $\QL_{EB}$ and the prepare-and-measure version of the protocol implemented by $\QL_{PM}$).
            
            \begin{theorem} \label{thm:EB=PM}
                Under Assumptions \ref{ass:Urandom} and \ref{ass:Sealed_labs}, %
                $\QL_{PM} \approx_0 \QL_{EB}$.
            \end{theorem}
            
            \begin{proof}
            
Let us suppose that for any honest player $j\in\Ho$,
                \begin{equation} \label{eq:SD_PM = SD_EB}
                    \SD_{PM}^j \approx_0 \SD_{EB}^j.
                \end{equation}

           Due to the fact that the systems $\QL_{PM}$ and $\QL_{EB}$ are obtained by an identical construction, based on respectively $\SD_{PM}^\Ho$ and $\SD_{EB}^\Ho$ (see \Cref{eq:Pj_PM_def,eq:Pj_EB_def,eq:Ql_sys_def}), the theorem follows from composing onto \Cref{eq:SD_PM = SD_EB} the systems $\PP$ and $\C$. The remainder of the proof is devoted to prove \Cref{eq:SD_PM = SD_EB} for any given $j\in\Ho$.
            
            For any $j\in\Ho$, $\SD_{PM}^j$ and $\SD_{EB}^j$ have the same interfaces, inputs and outputs.
            We first notice the following:
            \begin{itemize}
                \item The actions of $\SD_{PM}^J$ and $\SD_{EB}^J$ are the same.
                \item Both $\SD_{PM}^1$ and $\SD_{EB}^1$ output random bits $(b_n^1)_{n\in[N]}$ and the quantum state $\bigotimes_{n\in[N]}\frac1{\sqrt2}( \ket{0} + i^{b_n^1}\ket{1})$
            \end{itemize}
            Hence, $\SD_{PM}^J \approx_0 \SD_{EB}^J$ and $\SD_{PM}^1 \approx_0 \SD_{EB}^1$. We now deal with the case where $1<j<J$.
                
                Consider a distinguisher which is given black-box access to a system $\S \in \{\SD_{PM}^j, \SD_{EB}^j\}$ and whose goal is to distinguish the two cases. Let $n\in[N]$ be a fixed round and  
                $P$ be the register of the single-qubit system input of player $j$ at round $n$ (i.e $\rho_{n, in}^j$).
                Without loss of generality, we can consider that the distinguisher holds a register $D$ that contains a purification of $P$.
                The state of the whole system can be written as \begin{equation}
                \ket{\psi^{in}}_{DP} = \alpha \ket{\psi_0}_D\ket{0}_P + \beta \ket{\psi_1}_D\ket{1}_P
                \end{equation}
                with $\alpha, \beta\in\mathds{C}$.
                If $\S$ is $\SD_{PM}^j$, then player $j$ picks uniformly random $b_n^j$ and $v_n^j$, and applies the unitary $Z^{v_n^j - \frac{b_n^j}{2}}$ to system $P$. The resulting state is described by
                \begin{equation}\label{eq:PM_psiout}
                    \ket{\psi^{out}_{PM(b_n^j,v_n^j)}}_{DP} = \alpha \ket{\psi_0}_D\ket{0}_{P} + i^{2v_n^j - b_n^j}\beta \ket{\psi_1}_D\ket{1}_{P}
                \end{equation}
                
                We show that this state is indistinguishable from the one obtained when $S$ is $\SD_{EB}^j$. In this case, a CNOT gate is applied using register $P$ as control and the target register, call it $Q$, contains a fresh qubit in the state $\ket{0}_Q$. The overall state after this operation is described by:
                \begin{align*}
                    &\ket{\psi^{out}_{EB}}_{DPQ} \\
                    &=
                    (\mathds{I}_D \otimes CNOT_{PQ})(\ket{\psi^{in}}_{DP} \otimes \ket{0}_Q)\\
                    &= \alpha \ket{\psi_0}_D\ket{0}_{P}\ket{0}_Q + \beta \ket{\psi_1}_D\ket{1}_{P}\ket{1}_Q\\
                    \numberthis \label{eq:EB_psiout_0}
                    &= \frac1{\sqrt2}\Big[\Big( \alpha\ket{\psi_0}_D\ket{0}_{P} + \beta \ket{\psi_1}_D\ket{1}_{P} \Big) \ket{+}_Q +\Big( \alpha\ket{\psi_0}_D\ket{0}_{P} - \beta \ket{\psi_1}_D\ket{1}_{P}\Big) \ket{-}_Q\Big]\\
                    \numberthis \label{eq:EB_psiout_1}
                    &= \frac1{\sqrt2}\Big[\Big( \alpha\ket{\psi_0}_D\ket{0}_{P} - i\beta \ket{\psi_1}_D\ket{1}_{P} \Big) \ket{+_i}_Q + \Big( \alpha\ket{\psi_0}_D\ket{0}_{P} + i\beta \ket{\psi_1}_D\ket{1}_{P}\Big) \ket{-_i}_Q\Big]
                \end{align*}
                The register $Q$ is then measured in either the Hadamard or the circular basis, depending on the uniformly random bit $b_n^j$.
                From \Cref{eq:EB_psiout_0} and \Cref{eq:EB_psiout_1}, one can see that for all $b_n^j$ and $\ket{\psi^{in}}_{DP}$, 
                the outcome of the measurement is an uniformly random bit $v^j_n$. Moreover, the post-measurement state on registers $D$ and $P$ is exactly $\ket{\psi^{out}_{PM(b_n^j,v_n^j)}}_{DP}$%

                In conclusion, the outputs
                of round $n$ of $\SD_{PM}^j$ and $\SD_{EB}^j$ are exactly the same, and therefore indistinguishable, which proves \Cref{eq:SD_PM = SD_EB} for any $j\in\Ho$.
          \end{proof}

        \subsection{Security of $\QL_{EB}$} \label{sec:EB_security} %
\label{sec:security-eb}
      
            In this section, we prove the security of the entanglement based version of the protocol described in \Cref{sec:EB_state_distr}.
            
            \begin{theorem} \label{thm:EB_security}
                Under Assumptions~\ref{ass:2_pl_honest},~\ref{ass:Urandom},~\ref{ass:Sealed_labs}, and~\ref{ass:CAC},
                \begin{equation} \label{eq:eb_indist}
                    \QL_{EB}
                    \approx_{\eps} 
                    \I\circ\SIM
                \end{equation}
                where $\eps = (H-1)\eps_\QKD + \eps_{cor} $ with $\eps_\QKD$ the distinguishing advantage of a $QKD$ protocol with the same assumptions, and $\eps_{cor} = 2^{-\eta}$ the correctness parameter defined in \Cref{sec:correctness}
            \end{theorem}
    
            In order to prove \Cref{thm:EB_security}, we consider an unbounded distinguisher $\D$ which is given a system $\S$, either equal to $\QL_{EB}$ or to $\I\circ\SIM$, uniformly at random. 
            To avoid any ambiguities, we denote the input and outputs of $\QL_{EB}$ by 
            $
                \rho^\Ho_{in},
                t^{[J]\bs\Ho},
                \ell
            $
            and
            $
                s^\Ho,
                \rho^\Ho_{out}, 
                t %
            $
            , the input and outputs of $\I\circ\SIM$ by 
            $
                \tilde \rho^\Ho_{in},
                \tilde t^{[J]\bs\Ho},
                \tilde \ell
            $
            and
            $
                \tilde s^\Ho,
                \tilde \rho^\Ho_{out}, 
                \tilde t %
            $
            , and the input and outputs of $\S$ by 
            $
                \rho^\Ho_{\S, in},
                t^{[J]\bs\Ho}_\S,
                \ell_\S
            $
            and
            $
                s^\Ho_\S,
                \rho^\Ho_{\S, out}, 
                t_\S %
            $
            \footnote{This notation consists simply of \emph{labels} for the inputs and outputs of the abstract systems. They do not denote the underlying quantum states (that could be potentially entangled).}.

            \begin{lemma} \label{cor:eb_equiv_abort}
                The probability for $\S$ to abort is the same regardless of whether $\S$ is $\QL_{EB}$ or $\I\circ\SIM$.
                Moreover, conditioned on the event that $\S$ aborts, then
$                    \QL_{EB}
                    \approx_0
                    \I\circ\SIM $
            \end{lemma}

            \begin{proof}
                Note that in both $\QL_{EB}$ and $\I\circ\SIM$, 
                all the inputs are given to a $\QL_{EB}$ system, either directly in the real experiment or forwarded by $\SIM$ in the simulated one. This $\QL_{EB}$ system thus produces outputs with the exact same distribution in both cases. This gives the following:
                \begin{itemize}
                    \item Since $\I\circ\SIM$ aborts if and only if its simulation of $\QL_{EB}$ aborts, the first part of the lemma trivially holds.
                    \item If the systems abort, all the outputs of $\QL_{EB}$ and $\I\circ\SIM$ follow the exact same distribution. The shares $s^\Ho$ and $\tilde s^\Ho$ are indeed all set to $\bot$, while the other outputs directly come from the $\QL_{EB}$ system. This gives the second part of the lemma
                \end{itemize}
            \end{proof}

            \Cref{cor:eb_equiv_abort} allows us to focus on the case where the protocol does not abort. 

            Our goal is now to show that conditioned on $\S$ successfully terminating, the systems $\QL_{EB}$ and $\I\circ\SIM$ are indistinguishable.
            In order to prove this, we introduce an intermediate protocol that we call $\QKD$. The objective here is that $\QKD$ be close enough to standard entanglement-based QKD so that its security follows trivially from historical results, yet slightly modified so that we can reduce the security of our protocol to the security of $\QKD$

            In order to define $\QKD$, we first briefly recall entanglement-based QKD (EB-QKD).\footnote{We refer the reader to \cite{TL17} (part I) for a detailed description of EB-QKD.} 
            EB-QKD involves two players, Alice{}$_{QKD}$ and Bob{}$_{QKD}$ and works as follow:
            \begin{enumerate}
                \item \textbf{(State distribution)} Alice{}$_{QKD}$ and Bob{}$_{QKD}$ receive $N$ qubits and measure each qubit either in the computational or in the Hadamard basis, uniformly at random.
                \item \textbf{(sifting)} Alice{}$_{QKD}$ and Bob{}$_{QKD}$ announce their basis choices and discard the rounds where they did not agree.
                \item \textbf{(Error estimation)} Alice{}$_{QKD}$ and Bob{}$_{QKD}$ agree on a random subset of the remaining qubits and announce their measurement outcomes for the positions in that subset. They count the outcomes for which they disagree among this subset and abort the protocol if the corresponding rate is above a given threshold.
                \item \textbf{(Error correction)} Alice{}$_{QKD}$ announces the syndrome (relative to a given syndrome decoding protocol) of her secret and Bob{}$_{QKD}$ corrects his secret using this syndrome. They then compare hashes of the corrected secrets, aborting upon any mismatch.
                \item \textbf{(Privacy amplification)} Alice{}$_{QKD}$ and Bob{}$_{QKD}$ compute their final key as the image of their error-corrected measurement outcomes (the raw key) under a 2-universal hash function.
            \end{enumerate}
            
            \noindent We define the $\QKD$ protocol with the following modifications from EB-QKD.   %
            \begin{enumerate}[label=(\arabic*)]
                \item \label{item:diff_QKD'_1} In the state distribution step, for each position $n\in[N]$, Bob{}$_\QKD$ receives from the eavesdropper two bits $b_n^{(\D)}$ and $v_n^{(\D)}$. Instead of measuring in a basis given by $b_n^{\text{Bob}}$, he measures in basis $\hat b_n^{\text{Bob}}$ where
                \begin{equation} \label{eq:hatb_def}
                    \hat b_n^{\text{Bob}} = b_n^{\text{Bob}} \oplus b_n^{(\D)}
                \end{equation} 
                Additionally, Bob{}$_\QKD$ computes $\hat v_n^{\text{Bob}}$ where
                \begin{equation} \label{eq:hatv_def}
                    \hat v_n^{\text{Bob}} = v_n^{\text{Bob}} \oplus v_n^{(\D)} \oplus (b_n^{\text{Bob}} \vee b_n^{(\D)})
                \end{equation} 

                For the rest of the protocol, Bob{}$_\QKD$ uses $\hat b_n^{\text{Bob}}$ and $\hat v_n^{\text{Bob}}$ instead of $b_n^{\text{Bob}}$ and $v_n^{\text{Bob}}$.

                We notice that these two steps are equivalent to Bob applying  $\Z^{(v_n^{(\D)}+\frac12b_n^{(\D)})\pi}$ on his qubit before measuring in the $b^{\text{Bob}}_n$ basis\footnote{In the formula of $\hat v_n^{\text{Bob}}$ (\Cref{eq:hatv_def}), $v_n^{(\D)}$ flips the results according to the $\Z^{v_n^{(\D)}}$ component of the operation, and $(b_n^{j'} \vee b_n^{(\D)})$ accounts for the effect of the $\Z^{\frac12b_n^{(\D)}}$ part.}.

                \item \label{item:diff_QKD'_2} Before sifting, Alice{}$_\QKD$ and Bob{}$_\QKD$ agree on a random subset of all the qubits received in the state distribution step. At the error estimation step, this subset, restricted to the rounds that were not discarded during sifting, is used for the error rate computation instead of a freshly generated one.
                \item \label{item:diff_QKD'_3} In the error correction step, Bob{}$_{\QKD}$ does not correct his secret and instead announces its syndrome like Alice{}$_{\QKD}$ does.
            \end{enumerate}

            We claim now that Alice's key is as secure in $\QKD$ as in standard QKD. 
            To formalize this statement, we define the following "mask" systems, that replace or simply remove access for a given distinguisher to a specific output:
            \begin{itemize}
                \item $\M^\I_{Alice}$ takes as input the key of Alice{}$_\QKD$, and outputs instead a freshly generated, uniformly random bit string of the same length.
                \item $\M_{Bob}$ takes as input the key of Bob{}$_\QKD$, and has no output.
            \end{itemize} 
            Our claim amounts to the following \Cref{prop:QKD'_sec}, of which we defer the proof to \Cref{app:proof_QKD'sec}.

            \begin{proposition}\label{prop:QKD'_sec}
                \begin{equation}
                    \QKD\circ\M_{Bob} \approx_{\eps_{QKD}} \QKD\circ\M_{Bob}\circ\M^\I_{Alice}
                \end{equation}
                where $\eps_{QKD}$ is the distinguishing advantage of a $QKD$ protocol with the same parameters and under the same assumptions.
                
            \end{proposition}

            Similarly as above, we define the system $\M_{H}$ that takes as input the share $s^{j_H}$ and has no output.
            Using \Cref{prop:QKD'_sec}, we now prove \Cref{lem:indist_without_j_star}.

            \begin{lemma}\label{lem:indist_without_j_star}
                Let $\epsilon_\QKD$ be the distinguishing advantage of $\QKD$ when run with the same parameters as $\S$.
                Then, assuming $\S$ successfully terminates,

                \begin{equation} \label{eq:indist_without_j_star}
                    \QL_{EB} \circ \M_H
                    \approx_{(H-1) \eps_\QKD} 
                    \I\circ\SIM \circ \M_H
                \end{equation}
                
            \end{lemma}

            \begin{proof}
            Without loss of generality, the output of the distinguisher (on whether $S$ is $\QL_{EB}$ or $\I\circ\SIM$) is the outcome of a given measurement on the state $\sigma_\D$ containing all the outputs of $\S$ and the private register of the distinguisher.
            In particular, for all honest players $j\in\Ho$, we denote $S^j$ the register corresponding to the share output $s^j$.
            We denote $\sigma_{D\bs S^{j_H}} = Tr_{S^{j_H}}(\sigma_{D})$.
            
            In the remainder of the proof, we will aim to show that for any $h\in[H-1]$,
            \begin{equation} \label{eq:single-share-eps_states}
               \left(\frac{1}{2^{K}}I\right)^{\otimes h-1}\otimes Tr_{S^{j_{[h-1]}}}( \sigma_{D\bs S^{j_H}} ) \approx_{\eps_\QKD} \left(\frac{1}{2^{K}}I\right)^{\otimes h}\otimes Tr_{S^{j_{[h]}}}( \sigma_{D\bs S^{j_H}} ).
            \end{equation}

We notice that this finishes the proof, since chaining \Cref{eq:single-share-eps_states} for each $h \in [H-1]$, we have
            \begin{equation} \label{eq:combined-shares_states}
                \sigma_{D\bs S^{j_H}}
                \approx_{(H-1)\eps_\QKD}
                \left(\frac{1}{2^{K}}I\right)^{\otimes H-1}\otimes Tr_{S^\Ho}(\sigma_{D}),
            \end{equation}
which exactly gives \Cref{eq:indist_without_j_star}.

To conclude the proof of \Cref{lem:indist_without_j_star}, we prove \Cref{eq:single-share-eps_states} by contradiction. Let us suppose that there exist $h\in[H-1]$ and a distinguisher $\D_\QL$ which distinguishes $(\frac{1}{2^{K}}I)^{\otimes h-1}\otimes Tr_{S^{j_{[h-1]}}}( \sigma_{D\bs S^{j_H}} )$ from $(\frac{1}{2^{K}}I)^{\otimes h}\otimes Tr_{S^{j_{[h]}}}( \sigma_{D\bs S^{j_H}} )$ with probability $\eps_\D>\eps_\QKD$.

            Using this assumption, we will design an attack on $\QKD$ to contradict \Cref{prop:QKD'_sec}, thus proving \Cref{eq:single-share-eps_states}.
            The idea is to build a Qline around Alice{}$_\QKD$ and Bob{}$_\QKD$ such that when they perform $\QKD$, they are in fact taking part (as players $j_h$ and $j_H$ respectively) in a $\QL_{EB}$ protocol that we can then attack with $\D_\QL$. Formally, we will define a system $\A$ represented in \Cref{fig:A} which interfaces between the $\QKD$ protocol of Alice{}$_\QKD$ and Bob{}$_\QKD$ at an \emph{inner} interface, and the distinguisher $\D_\QL$ at an \emph{outer} interface.
            \begin{figure}[!ht]
    \centering \hspace*{0cm}
    \includegraphics[scale=0.14]{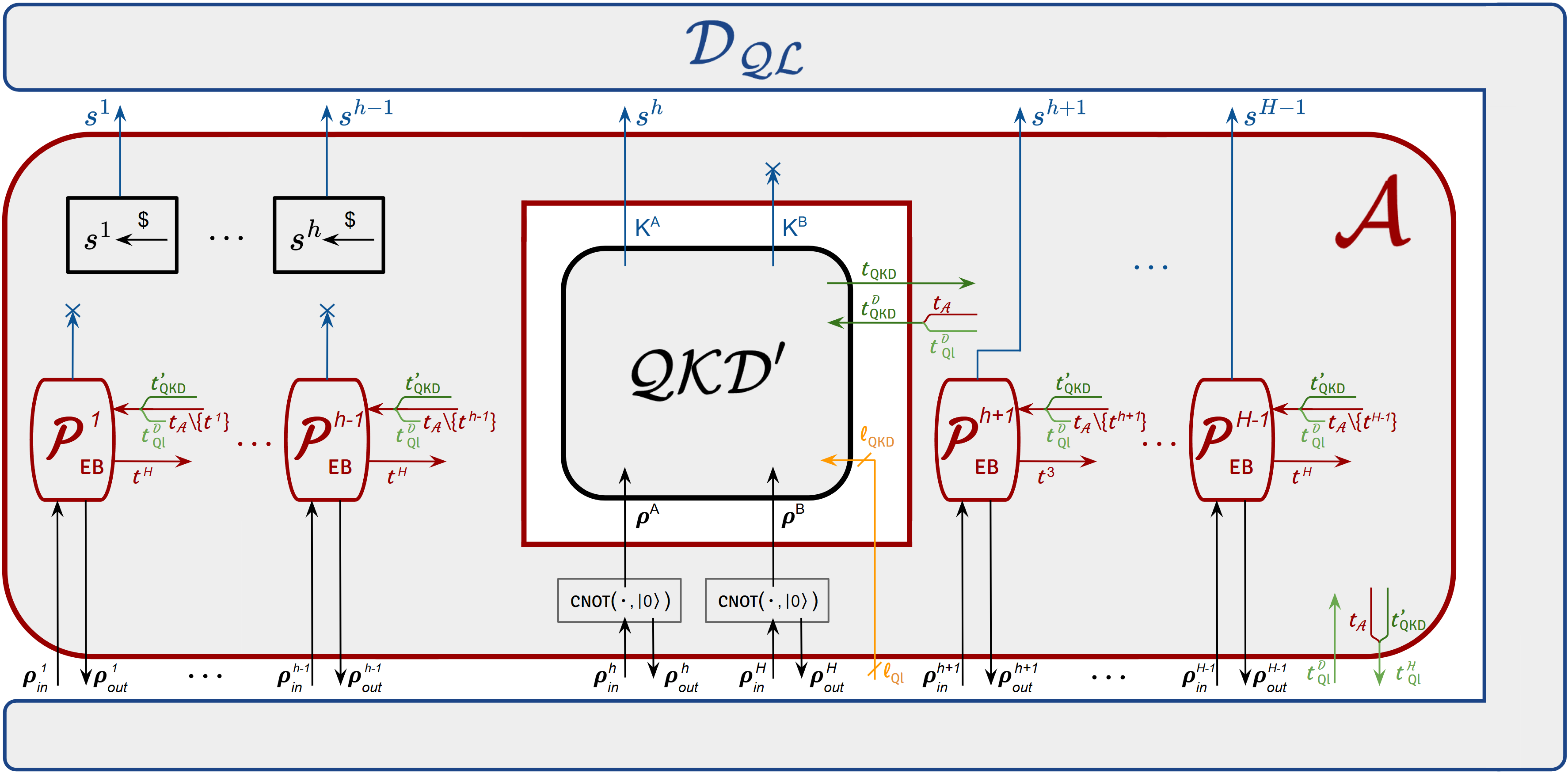}
    \caption{\centering The system $\A$ interfacing between $\QKD$ and $\D_\QL$. }
    \label{fig:A}
\end{figure}
            We define below the construction of $\A$ in details:\\
            $\A$ simulates the honest player systems $\P^{j}_{EB}$ of $\QL_{EB}$ for $j\in\Ho\bs\{j_h, j_H\}$.
                \begin{enumerate}
                    \item During the state distribution step:
                    \begin{itemize}
                        \item $\A$ directly forwards the quantum inputs/outputs $\rho_{in}^{j}$/$ \rho_{out}^{j}$ of $\P^{\Ho\bs\{j_h, j_H\}}_{EB}$ from/to the distinguisher at the outer interface of $\A$.
                        \item \label{item:A_def1} $\A$ interfaces the quantum inputs of Alice{}$_\QKD$ and Bob{}$_\QKD$ with the quantum inputs \emph{and outputs} of honest players $j_h$ and $j_H$ at its outer interface. Note that when performing $\QKD$, Alice{}$_\QKD$ and Bob{}$_\QKD$ exactly follow the state distribution step of $\QL_{EB}$ except for the step where the players propagate entanglement. $\A$ thus simulates this step by receiving at it's outer interface from the distinguisher $\D_\QL$ the qubits of input $\rho_{in}^{j_h}$, applying a CNOT gate on that qubit together with a freshly generated $\ket{0}$ and then sending one qubit to Alice{}$_\QKD$ and the second one to $\D_{\QL}$ via the outer output $\rho_{out}^{j_h}$ of $\A$. $\A$ performs similar steps for $j_H$ and Bob{}$_\QKD$.
                    \end{itemize}
                    \item During the post-processing step:
                    \begin{itemize}
                        \item $\A$ maps the blocking lever $l$ of its outer interface to the corresponding input of the authenticated channel of $\QKD$ 
                        \item The simulated honest player systems $\P^{\Ho\bs\{j_h, j_H\}}_{EB}$ comply with Alice{}$_\QKD$'s and Bob{}$_\QKD$'s choices for $\T'$, $\nu$, the syndrome decoding protocol, $f_{cc}$, $K$, and $f_{pa}$.
                        \item The inputs $b_{[N]}^{(\D)}$ and $v_{[N]}^{(\D)}$ of Bob{}$_\QKD$, provided by the eavesdropper in $\QKD$ protocol are defined by:
                        \begin{align*}
                            b_{n}^{(\D)} &= \bigoplus\limits_{j\in[J]\bs\{j_h, j_H\}} b_{n}^j\\
                            2v_{n}^{(\D)} + b_{n}^{(\D)} &= \sum\limits_{j\in[J]\bs\{j_h, j_H\}} 2v_{n}^j + b_{n}^j \mod4
                        \end{align*}
                        where the values $b_{[N]}^{[J]\bs\{j_h, j_H\}}$ and $v_{[N]}^{[J]\bs\{j_h, j_H\}}$ are defined either by the outputs of the simulated player systems $\P^j_{EB}$ for honest players $j\in\Ho\bs\{j_h, j_H\}$, or by $\D_\QL$ in the announcements of the dishonest players for $j\in[J]\bs\Ho$. Note that these dishonest announcements $v_{[N]}^{[J]\bs\Ho}$ are all well defined and accessible to $\A$ at this moment because of the random subset broadcast of $\QL_{EB}$ (see \Cref{ass:CAC}) which guarantees that the distinguisher $\D_\QL$ does not use the honest players' announcements to produces the dishonest players' ones. In other words, as $\A$ simulates the whole Qline which contains the authenticated channel with random subset broadcast, it has control over that channel and can compute some announcements based on others.
                        \item The simulated honest player systems $\P^{\Ho\bs\{j_h, j_H\}}_{EB}$ receive:
                        \begin{itemize}
                            \item The announcements of the other simulated honest players.
                            \item The announcements of the dishonest players that come from the distinguisher at the outer interface of $\A$.
                            \item The announcements of Alice{}$_\QKD$ labeled as coming from $j_h$.
                            \item The announcements of Bob{}$_\QKD$, labeled as coming from $j_H$ and slightly modified in order for the announcements that depend on $\hat v_n^{\text{Bob}}$ to instead match with what they would have been if Bob was using $v_n^{\text{Bob}}$ all along. This amounts to computing the modification string 
                            \begin{equation} \label{eq:vA_def}
                                v_n^{\A} = \hat v_n^{\text{Bob}} \oplus v_n^{\text{Bob}} = v_n^{(\D)} \oplus (b_n^{\text{Bob}} \vee b_n^{(\D)})
                            \end{equation} 
                            for all $n\in[N]$ (see \Cref{eq:hatv_def}), and to XOR the corresponding check bits announcement $v_\T^{\A}$, syndrome $w^{\A}$, correctness check output $c^\A$ and final share $s^{\A}$ to the respective announcements of Bob{}$_\QKD$.
                        \end{itemize}
                        \item $\A$ sends to $\D_\QL$ at it's outer interface the following announcements:
                        \begin{itemize}
                            \item The announcements of the simulated honest player systems $\P^{\Ho\bs\{j_h, j_H\}}_{EB}$.
                            \item The announcements of Alice{}$_\QKD$ labeled as coming from $j_h$.
                            \item The announcements of Bob{}$_\QKD$, labeled as coming from $j_H$ and modified the exact same way as described above.
                            
                        \end{itemize}
                        \item $\A$ exposes the following share outputs at it's outer interface:
                        \begin{itemize}
                            \item The shares $s^{j_{[h-1]}}$ are all sampled uniformly at random.
                            \item The share $s^{j_h}$ is directly the final key output of Alice{}$_\QKD$.
                            \item The shares $s^{j_{\{h+1,...,H-1\}}}$ are directly the output shares of the simulated honest players $\P^{j_{\{h+1,...,H-1\}}}_{EB}$.
                        \end{itemize}

                    \end{itemize}

                \end{enumerate}
                
                By connecting $\A$ onto $\D_\QL$ and $\QKD$ as depicted in \Cref{fig:A}, Alice{}$_\QKD$ and Bob{}$_\QKD$ are performing $\QKD$, while in the same time $\QKD \circ \A$ is undergoing the $\QL_{EB}$ secret sharing protocol.
                To obtain the contradiction, we use the following \Cref{prop:reduction} which we prove in \Cref{app:proof_QLD}. 
                
                \begin{proposition} \label{prop:reduction}
                    The system $\A\circ\D_\QL$ distinguishes $\QKD\circ\M_{Bob}$ from $\QKD\circ\M_{Bob}\circ\M^\I_{Alice}$ with probability $\eps_\D$.
                \end{proposition}
                
                As  $\eps_\D > \eps_\QKD$, this contradicts \Cref{prop:QKD'_sec}, hence proving \Cref{eq:single-share-eps_states}.
        \end{proof}

            \begin{lemma} \label{lem:eb_equiv_success}
                If $\S$ successfully terminates, then
                \begin{equation} \label{eq:main_indist}
                    \QL_{EB}
                    \approx_\eps
                    \I\circ\SIM
                \end{equation}
                with $
                    \eps = \eps_{cor} + (H-1) \eps_\QKD
                $
                where $\eps_\QKD$ is the distinguishing advantage of Alice's key in a $\QKD$ protocol with the same parameters and under the same assumptions.
            \end{lemma}
\begin{proof}
We notice that following the description of the protocol, $s^{[J]}$ is a function of the values $v^{[J]}_{[N]}$ of the players along with publicly known data. As these values are all input in the random subset broadcast functionality of $\S$, they are all well defined in $\S$ and thus so are in particular the dishonest player's expected shares that we denote $s_\S^{[J]\bs\Ho}$.

We assume that correctness (\Cref{eq:correctness}) holds for the $\QL_{EB}$ system of $\S$. From \Cref{prop:correctness}, this happens except with probability $\eps_{cor}$.

If $\S$ is $\I\circ\SIM$, by definition of $\I$ and $\SIM$ (see \Cref{eq:sH,eq:scompr}), we have
\begin{align*}
     s^{j_H}
     & = \Big(\bigoplus_{j \in \Ho\bs\{j_H\}} \tilde s^{j}\Big) \oplus s^{compr}\\
     &= \Big(\bigoplus_{j \in \Ho\bs\{j_H\}} \tilde s^{j}\Big) \oplus \Big(\bigoplus\limits_{j\in \Ho} s^j_\SIM\Big)
\end{align*}
which as we assume correctness gives 
\begin{equation*}
     s^{j_H} = \Big(\bigoplus_{j \in \Ho\bs\{j_H\}} \tilde s^{j}\Big) \oplus \Big(\bigoplus\limits_{j\in [J]\bs\Ho} s^j_\S\Big)
\end{equation*}

Similarly, if $\S$ is $\QL_{EB}$, by correctness,
\begin{align*}
     s^{j_H} &= \bigoplus_{j \in [J]\bs\{j_H\}} s^{j}\\
     &\Big(\bigoplus_{j \in \Ho\bs\{j_H\}} s^{j}\Big) \oplus \Big(\bigoplus\limits_{j\in [J]\bs\Ho} s^j_\S\Big)
\end{align*}

In conclusion, in both cases, $s^{j_H}$ is identically defined from the other inputs and outputs of $\S$ as
\begin{equation*}
     s^{j_H} = \Big(\bigoplus_{j \in \Ho\bs\{j_H\}} s_\S^{j}\Big) \oplus \Big(\bigoplus\limits_{j\in [J]\bs\Ho} s^j_\S\Big)
\end{equation*}

As a result, except with probability $\eps_{cor}$, \Cref{lem:indist_without_j_star} holds even when giving the distinguisher access to $s^{j_H}$, meaning even when removing the $\M_H$ systems from \Cref{eq:indist_without_j_star}. This concludes the proof  
\end{proof}

            \Cref{thm:EB_security} then follows from \Cref{cor:eb_equiv_abort} and \Cref{lem:eb_equiv_success}. 
            
        \subsection{Conclusion on the security}
\label{sec:conclusion-security}
        Combining the main results of the previous sections (\Cref{thm:EB_security} and \Cref{thm:EB=PM}) using the composability of the security definition, we obtain the following final statement.

        \begin{theorem}\label{thm:PM_security}
        Under Assumptions~\ref{ass:2_pl_honest},~\ref{ass:Urandom},~\ref{ass:Sealed_labs} and~\ref{ass:CAC}, the following holds. 
        \begin{equation}
            \QL_{PM} \approx_\eps \I \circ \SIM
        \end{equation}
        with 
        \begin{equation}
            \eps = \eps_{cor} + (H-1) \eps_{QKD}
        \end{equation}
        \end{theorem}

        An example for an expression of $\eps_{QKD}$ can be obtained from \cite{TL17} (Theorem 3, Equations (57) and (58)), yielding the following, where $\chi$ is the size of the syndromes involved in the error correction step and $h(x)= -xlog(x) - (1-x)log(1-x)$ is the binary entropy function.
        \begin{equation} %
            \eps_\QKD = %
            2e^{-\frac{M\tau^2}{L(\tau+1)}\nu^2} + \frac12 \sqrt{2^{-M\big(1-h(\delta + \nu)\big) + \eta + \chi + K }}
        \end{equation}

        For two honest players, our bound for $\eps$  matches known bounds for standard QKD. This is natural since the security proof follows from a reduction to QKD.
        However, the bound worsens when the number of honest players grows. This counter-intuitive behavior could be explained by the fact that more honest players imply more potential targets (shares) to distinguish from uniformly random ones. The composability of our security definition indeed imposes, in order for the whole scheme to be secure, that all shares are secure together, and not only individually.

\newpage
\bibliographystyle{ieeetr}
\bibliography{MyRef.bib}

\newpage
\appendix

\section{Notation} \label{sec:notations}

\begin{table}[!h]
\begin{tabularx}{\textwidth}{p{0.14\textwidth}X}

  \multicolumn{2}{l}{\ul{Mathematical notation:}} \vspace{1mm} \\
  $\oplus$ & Addition modulo $2$. \\
  $\vee$ & Logical "or" between binary variables. \\
  $[N]$ & The set $\{1,...,N\}$ of integers ranging from 1 to $N$. \\
  $|A|$ & Size of the set $A$. \\
  $log$ & Base 2 logarithm function of strictly positive real numbers.\\ 
  $h(x)$ & Binary entropy function of $x\in]0, 1[$. $h(x)= -xlog(x) - (1-x)log(1-x)$ \\ 
  $\R \circ \S$ & The sequential composition of two systems. \\ 
  $\R || \S$ & The parallel composition of two systems. \\
  $\approx_\eps$ & Indistinguishability, except with probability $\eps$. \quad Between abstract systems, $R\approx_\eps S$ is relative to the distinguishing pseudo-metric, while between quantum states $\sigma \approx_\eps \gamma$ is relative to the trace distance.\\
  $Tr (A)$ & The trace of a matrix $A$. \\
  $\ket{+}$, $\ket{-}$ & Eigenstates of the Hadamard basis. \\
  $\ket{+_i}$, $\ket{-_i}$ & Eigenstates of the circular basis. \\
  $\Z$ & The Pauli $\Z$ single-qubit gate
  $\begin{bmatrix}
      1&0\\
      0&-1
  \end{bmatrix}$. $\Z^\frac 12$ is the phase gate 
  $\begin{bmatrix}
      1&0\\
      0&i
  \end{bmatrix}$ \vspace{4mm} \\ 
    
  \multicolumn{2}{l}{\ul{Abstract systems and protocols:}}    \vspace{1mm} \\
  $\C$ & Classical authenticated channel with random subset broadcast subroutine representing \Cref{ass:CAC}. \\
  $\SD^j_{PM}, \SD^j_{EB}$ & State distribution system of player $j$ for the prepare-and-measure and entanglement-based versions of the protocol respectively. \\
  $\PP$ & Post-processing system. \\
  $\P^j_{PM}, \P^j_{EB}$ & Player $j$'s whole system for the prepare-and-measure and entanglement-based versions of the protocol respectively. $\P^j_{PM} = \PP \circ \SD^j_{PM}$ and $\P^j_{EB} = \PP \circ \SD^j_{EB}$ \\
  $\QL_{PM}, \QL_{EB}$ & Whole system of an execution of the prepare-and-measure and entanglement-based versions of the protocol respectively.\\
  $\I$ & Ideal system for the secret sharing primitive. \\
  $\SIM$ & Simulator (see \Cref{fig:Isigma}). \\
  $\D$ & Distinguisher which is given $\S$ and tries to guess which of the two possible systems it is. \\
  $\S$ & Abstract system given to the Distinguisher. \\
\end{tabularx}
\end{table}

\begin{table}[!h]
\begin{tabularx}{\textwidth}{p{0.08\textwidth}X}

  \multicolumn{2}{l}{\ul{Protocol parameters:}}   \vspace{1mm} \\
  $\eps$ & Indistinguishability parameter (Security parameter). \\
  $\tau', \tau$ & Sizes of the sets $\T'$ and $\T$ respectively. \\
  $q$ & Qubit error rate. \\
  $\delta$ & Threshold of the qubit error rate $q$ for the protocol to proceed after error estimation. \\
  $\nu$ & Error correction margin. It is to be chosen by the players to adjust the final share size $K$ and security parameter$\eps$. \\
  $\chi$ & Bit length of the syndromes $(w^j)_{j\in[J]}$.\\
  $\eta$ & Correctness parameter and binary size of the output of $f_{cc}$. \\
  $f_{cc}$ & Correctness check hash function mapping $\{0, 1\}^{M}$ to $\{0, 1\}^{\eta}$. \\
  $f_{pa}$ & Privacy amplification hash function mapping $\{0, 1\}^{M}$ to $\{0, 1\}^{K}$. \vspace{5mm} \\
  
  \multicolumn{2}{l}{\ul{Indices and Sets:}}  \vspace{1mm} \\
  $J$ & Total number of players ($j\in[J]$). \\ 
  $\Ho$ & The set of honest players ($j\in\Ho$). The elements of $\Ho$ are denoted $j_1, ..., j_H$. \\ 
  $H$ & Number of honest players. $H = |\Ho|$ \\ 
  $N$ & Number of rounds of the protocol ($n\in[N]$).  \\
  $L$ & Number of remaining data bits after sifting. $L = N-|\U|$ \\
  $M$ & Number of remaining data bits after error estimation. $M = L-\tau$ \\
  $K$ & Bit size of the final shares. \\
  $\T', \T$ & Set of indices of the check bits, before and after sifting respectively. \\
  $\U$ & Set of indices of the rounds discarded for uncorrelated basis choices. $\U\subset[N]$ \vspace{5mm} \\
  
  \multicolumn{2}{l}{\ul{Other notation:}}  \vspace{1mm} \\
  $QKD$ & Quantum Key distribution protocol. \\
  $\QKD$ & A slightly modified QKD protocol introduced \Cref{sec:EB_security} and used as a tool for the security proof\\
  $j_1, ..., j_H$ & The elements of $\Ho$ in their corresponding order on the Qline (meaning  $j_1< ...<j_H$)\\
  $b_n^j, b^j$ & $b_n^j$: Basis bit of player $j$ for round $n$ in $\QL_{EB}$.\quad\hspace{0.1mm} $b^j = (b_n^j)_{n\in[N]}$. \\
  $v_n^j, v^j$ & $v_n^j$: Value bit of player $j$ for round $n$ in $\QL_{EB}$.\quad\hspace{0.1mm} $v^j = (v_n^j)_{n\in[N]}$ \\
  $w^j$ & Syndrome of the raw key $(v_n^j)_{n\in[M]}$ of player $j$ announced during the error correction step \\ %
  $s^j$ & Final share of player $j$. Share output for player $j$ in systems performing secret sharing. \\
  $\rho_{n, in}^j$ & Single-qubit quantum state input of player $j$ for round $n$. $\rho_{in}^j = (\rho_{n, in}^j)_{n\in[N]}$ \\ 
  $\rho_{n, out}^j$ & Single-qubit quantum state output of player $j$ for round $n$. $\rho_{out}^j = (\rho_{n, out}^j)_{n\in[N]}$ \\
  $t$ & Whole transcript of all the classical authenticated communications that occur during the protocol. \\
  $\ell$ & Blocking lever. \\
  $\rho^\Ho_{in}, \rho^\Ho_{out}, t^{[J]\bs\Ho}, \ell, s^\Ho$ & \quad\quad\quad\quad\quad\quad\quad\quad inputs and outputs of $\QL_{EB}$. \\
  $\tilde \rho^\Ho_{in}, \tilde \rho^\Ho_{out}, \tilde t^{[J]\bs\Ho}, \tilde \ell, \tilde s^\Ho$ & \quad\quad\quad\quad\quad\quad\quad\quad inputs and outputs of $\I\circ\SIM$. \\
  $\rho^\Ho_{\S,in}, \rho^\Ho_{\S,out}, t^{[J]\bs\Ho}_\S, \ell_\S, s^\Ho_\S$ & \quad\quad\quad\quad\quad\quad\quad\quad  inputs and outputs of $\S$ (in \Cref{sec:EB_security}). \\

 \end{tabularx}
\end{table}

For any variable $x$ and any sets or families $A$ and $B$, $x^A_B$ denotes $(x^a_b)_{a\in A, b\in B}$.

\clearpage

\section{Proofs}

    \subsection{Proof of \Cref{prop:QKD'_sec}} \label{app:proof_QKD'sec}
        The security of entanglement-based QKD has been shown by several different works (\cite{TL17,Portman_Renner14,BF12}). More concretely, under Assumptions~\ref{ass:Urandom}, \ref{ass:Sealed_labs} and \ref{ass:CAC}, there exists $\eps_{QKD}$ negligible in $N$ such that the key obtained by Alice through this protocol is indistinguishable to a uniformly random bit string of the same length except with probability $\eps_{QKD}$. 
                        
        We discuss here why the three differences of $\QKD$ preserve this indistinguishability. 
        \begin{enumerate}
        
            \item First, we notice that difference~\ref{item:diff_QKD'_1} does not impact the security because the quantum communication channel in $QKD$ is assumed to be controlled by the adversary. Any attack on a $QKD$ protocol with such modification can be turned into an attack on $QKD$ by simply applying Bob{}$_\QKD$'s rotation to the input states of Bob{}$_{QKD}$.
            
            \item Difference~\ref{item:diff_QKD'_2} preserves the main properties of the subset for the error rate computation, which are that it is uniformly random (among all subsets of the same size) and independent of the secrets of the players. These are indeed sufficient properties to prove that the error estimation gives a good estimate of the amount of errors outside that subset \footnote{Such a proof can be found in \cite{TL17}, Proposition 8.}.
                 
            \item For the case of difference~\ref{item:diff_QKD'_3} we introduce some notation. Let $K_A$ and $K_B$ be the respective raw keys (after sifting and before error correction) of Alice and Bob in QKD. These keys are linked by the relation $K_A = K_B \oplus e$ where $e$ is the error string that Bob ought to identify during error correction. By the linearity of the syndrome function $w(\cdot)$ of the error correcting code, we have that $w(K_A) = w(K_B) \oplus w(e)$. Hence, During the error correction step, as the eavesdropper learns $w(K_A)$, difference~\ref{item:diff_QKD'_3}, which leaks $w(K_B)$ to the eavesdropper, is equivalent to leaking $w(e)$. The impact of the leakage of $w(e)$ (or equivalently of $e$) in QKD protocols has been studied in \cite{syndrome_leakage}. Under our assumptions, and because the protocol does not make use of more statistics than the qubit error rate during the parameter estimation step, this leak does not reveal any additional information on $K_A$ to the eavesdropper.
            As a consequence, difference~\ref{item:diff_QKD'_3} does not influence the security of $K_A$.
            
        \end{enumerate}

         In conclusion, none of the differences between standard QKD and $\QKD$ influences the indistinguishability of Alice's final key to a random bit string. This concludes the proof.

    \subsection{Proof of Proposition \ref{prop:reduction}} \label{app:proof_QLD} 
        Let $\M^\I_{[h-1]}$ and $\M^\I_{[h]}$ be systems that respectively take as input the shares $s^{j_{[h-1]}}$ and $s^{j_{[h]}}$ and outputs instead the same amount of random bit strings of the same size.
        With this notation, our assumption on the distinguisher amounts to the fact that $\D_\QL$ distinguishes $\QL_{EB}\circ (\M_H || \M^\I_{[h-1]})$ from $\QL_{EB}\circ (\M_H || \M^\I_{[h]})$ with probability $\eps_\D$.

        We wish to show that
        \begin{equation} \label{eq:A_QL}
             \QKD\circ\M_{Bob} \circ \A 
             \approx_0
             \QL_{EB}\circ (\M_H || \M^\I_{[h-1]}),
        \end{equation}
        and
        \begin{equation} \label{eq:A_I}
             (\QKD\circ\M_{Bob}\circ\M^\I_{Alice}) \circ \A  
             \approx_0 
             \QL_{EB}\circ (\M_H || \M^\I_{[h]}).
        \end{equation}

        We first focus on \Cref{eq:A_QL}.
        Note that in the two systems $\QKD\circ\M_{Bob}\circ\A $ and $\QL_{EB}\circ (\M_H || \M^\I_{[h-1]})$, the behavior of all players $j\in\Ho\bs\{j_h, j_H\}$ is the same.

        Furthermore, one can see that in both systems:
        \begin{itemize}
            \item At the state distribution step, the input state of players $j_h$ and $j_H$ undergoes identical CNOT gates with freshly generated qubits and is output back. The remaining qubits are then measured in random basis (for Bob, the basis are $b_{[N]}^{j_H}$ in $\QL_{EB}$ and $\hat b_{[N]}^{Bob}$ in $\QKD$).
            \item At the sifting step, in $\QL_{EB}$, a round $n$ is discarded if (see \Cref{eq:hatb_def})
            \begin{align*} 
                &b_n^{j_h} \oplus 
                b_n^{j_H} \oplus 
                \big[ \bigoplus\limits_{j \in [J]\bs \{j_h, j_H\}} b^j_n \big]
                \ne 0\\
            \numberthis \label{eq:EB_br_cond}
           & \Leftrightarrow \quad
               b_n^{j_h} \oplus 
                b_n^{j_H} \oplus 
                b^{(\D)}_n
                \ne 0\\
            &\Leftrightarrow \quad 
               b_n^{j_h} \ne \hat b_n^{Bob}
            \end{align*}
            which is the exact discarding condition of the sifting step of $\QKD\circ\M_{Bob}\circ\A$. 
            
            \item At the error estimation step in $\QL_{EB}$, a round $n$ is considered as erroneous if 
            \begin{align*}
                        &(2v_{n}^{j_h} + b_{n}^{j_h}) + (2v_{n}^{j_H} + b_{n}^{j_H}) + 
                        \big(\sum\limits_{j \in [J]\bs \{j_h, j_H\}} ( 2 v^j_{n} + b^j_{n} ) \big) &= 2 \mod 4 \\
                        \numberthis \label{eq:err_est_cond}
                       & \Leftrightarrow \quad2 (v_{n}^{j_h} + v_{n}^{j_H} + v_{n}^{(\D)}) + b_{n}^{j_h} +  b_{n}^{j_H} + b_{n}^{(\D)} &= 2 \mod 4.
            \end{align*}
            Note that because of the sifting step, all rounds satisfying \Cref{eq:EB_br_cond} have been discarded. As a consequence, for all ${n} \in [\T]$, $b_{n}^{j_h} = b_{n}^{j_H} \oplus b_{n}^{(\D)}$ and thus $b_{n}^{j_h} + b_{n}^{j_H} + b_{n}^{(\D)} = 2 (b_{n}^{j_H} \vee b_{n}^{(\D)})$\footnote{This is true for any binary $a$, $b$, $c$: $a=b\oplus c \Rightarrow  a+b+c = 2(b\vee c)$.}. Condition~{(\ref{eq:err_est_cond})} is then equivalent to (see \Cref{eq:hatv_def})
            \begin{align*}
                        & 2 \big(v_{n}^{j_h} + v_{n}^{j_H} + v_{n}^{(\D)} + (b_{n}^{j_H} \vee b_{n}^{(\D)}) \big) = 2 \mod 4 \\
                        &\Leftrightarrow \quad v_{n}^{j_h} + v_{n}^{j_H} + v_{n}^{(\D)} + (b_{n}^{j_H} \vee b_{n}^{(\D)}) = 1 \mod 2\\
                       & \Leftrightarrow\quad v_{n}^{j_H} \oplus v_{n}^{(\D)} \oplus (b_{n}^{j_H} \vee b_{n}^{(\D)}) \ne v_{n}^{j_h}\\
                      &  \Leftrightarrow\quad \hat v_{n}^{Bob} \ne  v_{n}^{j_h}
            \end{align*}
                
            which is exactly the condition under which rounds are considered erroneous in $\QKD\circ\M_{Bob}\circ\A$.
            \item Since the sifting and error estimation are similar, the announcements are made following the same process. Moreover, in both systems, they are based on the basis used to measure the state and the outcome of that measurement. This is indeed directly the case for $\QL_{EB}$. In $\QKD\circ\M_{Bob}\circ\A$, the announcements are those of Bob{}$_\QKD$ based on $\hat v_{[N]}^{Bob}$, later XOR-ed by $\A$ with announcements corresponding to $v_{[N]}^\A = \hat v_{[N]}^{Bob} \oplus v_{[N]}^{Bob}$ (see \Cref{eq:vA_def}). By linearity of the error correction and correctness check, the modified announcements exactly match the ones corresponding to the actual outcomes $v_{[N]}^{Bob}$ of the measurements of Bob{}$_\QKD$.
        \end{itemize}
        As a consequence, $\QKD\circ\M_{Bob}\circ\A $ and $\QL_{EB}\circ (\M_H || \M^\I_{[h-1]})$ are indistinguishable, which gives \Cref{eq:A_QL}.

        The same reasoning applies for \Cref{eq:A_I}, the only difference being the share outputs of players $j_{[h]}$ that are replaced by randomly sampled bit strings in both systems.

        Composing $\D_\QL$ onto \Cref{eq:A_QL,eq:A_I}, and using the property that $\D_\QL$ distinguishes $\QL_{EB} \circ (\M_H || \M^\I_{[h-1]})$ from $\QL_{EB} \circ (\M_H || \M^\I_{[h]})$ with probability $\eps_\D$, we get that the system $\A\circ\D_\QL$ distinguishes $\QKD\circ\M_{Bob}$ from $\QKD\circ\M_{Bob}\circ\M^\I_{Alice}$ with that same probability $\eps_\D$.
        This concludes the proof.

\end{document}